\documentclass{statsoc}
\usepackage[a4paper, top=1in, bottom=1in, left=1in, right=1in]{geometry}
\usepackage{amsmath}
\usepackage{amssymb}
\usepackage{enumerate}
\usepackage{graphicx}
\usepackage{color}
\usepackage{algorithm}
\usepackage{algpseudocode}
\usepackage{caption, subcaption}

\usepackage{natbib}

\newcommand{\cX}{\mathcal{X}}

\newcommand{\cH}{\mathcal{H}}
\newcommand{\cF}{\mathcal{F}}
\newcommand{\cG}{\mathcal{G}}
\newcommand{\cC}{\mathcal{C}}
\newcommand{\cN}{\mathcal{N}}
\newcommand{\cA}{\mathcal{A}}
\newcommand{\cR}{\mathcal{R}}

\newcommand{\hFDP}{\widehat{\textnormal{FDP}}}
\newcommand{\hfdr}{\widehat{\textnormal{fdr}}}
\newcommand{\hfdp}{\widehat{\textnormal{fdp}}}

\newcommand{\htt}{\hat{t}}

\newcommand{\tp}{\tilde{p}}

\newcommand{\FDP}{\textnormal{FDP}}
\newcommand{\FDR}{\textnormal{FDR}}
\newcommand{\fdr}{\textnormal{fdr}}
\newcommand{\adapt}{\textnormal{AdaPT} }
\newcommand{\update}{\textsc{update} }

\newcommand{\eps}{\epsilon}
\newcommand{\lb}{\left(}
\newcommand{\rb}{\right)}
\newcommand{\td}{\tilde}
\newcommand{\R}{\mathbb{R}}
\newcommand{\E}{\mathbb{E}}
\renewcommand{\P}{\mathbb{P}}

\newcommand{\tX}{\widetilde{X}}
\newcommand{\1}{\mathbf{1}}
\newcommand{\exps}[1]{\exp\left\{#1\right\}}
\newcommand{\pkg}[1]{\texttt{#1}}

\newtheorem{proposition}{Proposition}
\newtheorem{lemma}{Lemma}
\newtheorem{theorem}{Theorem}

\title{AdaPT: An interactive procedure for multiple testing with side information}
\author{Lihua Lei and William Fithian}
\address{Department of Statistics, University of California, Berkeley, USA}
\email{\{lihua.lei, wfithian\}@berkeley.edu}

\begin{document}

\maketitle

\begin{abstract}
  We consider the problem of multiple hypothesis testing with generic side information: for each hypothesis $H_i$ we observe both a $p$-value $p_i$ and some predictor $x_i$ encoding contextual information about the hypothesis. For large-scale problems, adaptively focusing power on the more promising hypotheses (those more likely to yield discoveries) can lead to much more powerful multiple testing procedures. We propose a general iterative framework for this problem, called the Adaptive $p$-value Thresholding (AdaPT) procedure, which adaptively estimates a Bayes-optimal $p$-value rejection threshold and controls the false discovery rate (FDR) in finite samples. At each iteration of the procedure, the analyst proposes a rejection threshold and observes partially censored $p$-values, estimates the false discovery proportion (FDP) below the threshold, and proposes another threshold, until the estimated FDP is below $\alpha$. Our procedure is adaptive in an unusually strong sense, permitting the analyst to use any statistical or machine learning method she chooses to estimate the optimal threshold, and to switch between different models at each iteration as information accrues. We demonstrate the favorable performance of \adapt by comparing it to state-of-the-art methods in five real applications and two simulation studies.
\end{abstract}


\section{Introduction}\label{sec:intro}

\subsection{Interactive data analysis}

In classical statistics we assume that the question to be answered, and the analysis to be used in answering the question, are both fixed in advance of collecting the data. Many modern applications, however, involve extremely complex data sets that may be collected without any specific hypothesis in mind. Indeed, very often the express goal is to explore the data in search of insights we may not have expected to find. A central challenge in modern statistics is to provide scientists with methods that are flexible enough to allow for exploration, but that nevertheless provide statistical guarantees for the conclusions that are eventually reported.

Selective inference methods blend exploratory and confirmatory analysis by allowing a search over the space of potentially interesting questions, while still guaranteeing control of an appropriate Type I error rate such as a conditional error rate \citep[e.g.,][]{yekutieli2012adjusted,lee2016exact,fithian2014optimal}, familywise error rate~\citep[e.g.,][]{tukey1994collected,berk2013valid}, or false discovery rate~\citep[e.g.,][]{bh95,barber2015controlling}. However, most selective inference methods require that the selection algorithm be specified in advance, forcing a choice between either ignoring any difficult-to-formalize domain knowledge or sacrificing statistical validity guarantees. 

Interactive data analysis methods relax the requirement of a pre-defined selection algorithm. Instead, they provide for an interactive analysis protocol between the analyst and the data, guaranteeing statistical validity as long as the protocol is followed. The two central questions in interactive data analysis are ``what did the analyst know and when did she know it?'' Previous methods for interactive data analysis involve randomization~\citep{dwork2015preserving, tian2015selective} to control the analyst's access to the data at the time she decides what questions to ask. 

This paper proposes an iterative, interactive method for multiple testing in the presence of {\em side information} about the hypotheses. We restrict the analyst's knowledge by partially censoring all $p$-values smaller than a currently-proposed rejection threshold, and guarantee finite-sample FDR control by applying a version of the optional-stopping argument pioneered by~\citet{storey04} and extended in \citet{barber2015controlling,gsell2016sequential,li2016accumulation,lei2016power,barber2016knockoff}.

\subsection{Multiple testing with side information}

In many areas of modern applied statistics, from genetics and neuroimaging to online advertising and finance, researchers routinely test thousands or millions or hypotheses at a time. For large-scale testing problems, perhaps the most celebrated multiple testing procedure of the modern era is the Benjamini--Hochberg (BH) procedure \citep{bh95}. Given $n$ hypotheses and a $p$-value for each one, the BH procedure returns a list of rejections or ``discoveries.'' If $R$ is the number of total rejections and $V$ is the number of false rejections (rejections of true null hypotheses), the BH procedure controls the  {\em false discovery rate} (FDR), defined as
\begin{equation}\label{eq:FDR}
\text{FDR} = \E\left[\frac{V}{\max\{R,1\}}\right],
\end{equation}
at a user-specified target level $\alpha$. The random variable $V/\max\{R,1\}$ is called the {\em false discovery proportion} (FDP). 

The BH procedure is nearly optimal when the null hypotheses are exchangeable {\em a priori}, and nearly all true. In other settings, however, the power can be improved, sometimes dramatically, by applying prior knowledge or by learning from the data. For example, adaptive FDR-controlling procedures can gain in power by estimating the overall proportion of true nulls \citep{storey02}, applying priors to increase power using $p$-value weights \citep{benjamini1997multiple, genovese2006false, dobriban2015optimal, dobriban2016general}, grouping similar null hypotheses and estimating the true null proportion within each group \citep{hu2012false}, or exploiting a prior ordering to focus power on more ``promising'' hypotheses near the top of the ordering \citep{barber2015controlling,gsell2016sequential, li2016accumulation, lei2016power}.

In most large-scale testing problems, the null hypotheses do not comprise an undifferentiated list; rather, each hypothesis is associated with rich contextual information that could potentially help to inform our testing procedures. For example, \citet{li2016accumulation} test for differential expression of 22,283 genes between a treatment and control condition for a breast cancer drug, with side information in the form of an ordering of genes from most to least ``promising'' using auxiliary data collected at larger dosages. Multiple testing procedures that exploit the ordering can reject hundreds of hypotheses while the BH procedure (which does not exploit the ordering) rejects none. 

More generally, prior information could arise in more complex ways. For example, consider testing for association of 400,000 single-nucleotide polymorphisms (SNPs) with each of 40 related diseases. If gene-regulatory relationships are known, then we might expect SNPs near related genes to be associated (or not) with related diseases, but without knowing ahead of time which gene-disease pairs are promising. In a similar vein, \citet{fortney2015genome} used prior knowledge of each SNP's associations with age-related diseases to focus their search for SNPs associated with longevity, leading to novel discoveries. Inspired by examples like this, \citet{ignatiadis2016data} and \citet{li2016multiple} have recently proposed a more general problem setting where, for each hypothesis $H_i$, $i\in [n]$ we observe not only a $p$-value $p_i\in [0,1]$ but also a predictor $x_i$ lying in some generic space $\cX$. Unlike $p_i$, $x_i$ carries only indirect information about the hypothesis: it is meant to capture some side information that might bear on $H_i$'s likelihood to be false, or on the power of $p_i$ under the alternative, but the nature of this relationship is not fully known ahead of time and must be learned from the data.

In other situations, the ``predictor'' information could simply represent a measure of sample size or overall signal for testing the $i$th hypothesis, which could be informative about the power of the $i$th test to distinguish the alternative from the null. For example, if each $p_i$ concerns a test for association between the $i$th SNP and a disease, then the overall prevalence of that SNP (in the combined treatment and control groups) can be used as prior information. Or, if $p_i$ arises from a two-sample $t$-test, we could use the pooled variance, the sample variance ignoring the group labels, as prior information; see e.g. \citep{bourgon2010independent, ignatiadis2016data}.

\subsection{AdaPT: a framework for FDR control}\label{subsec:framework}

This paper presents a new framework for FDR control with generic side information, which we call {\em adaptive $p$-value thresholding} or AdaPT for short. Our method proceeds iteratively: at each step $t=0,1,\ldots$, the analyst proposes a rejection threshold $s_t(x)$ and computes an estimator $\hFDP_t$ for the false discovery proportion for this threshold. If $\hFDP_t\leq \alpha$, she stops and rejects every $H_i$ for which $p_i\leq s_t(x_i)$. Otherwise, she proposes a more stringent threshold $s_{t+1} \preceq s_t$ and moves on to the next iteration, where the notation $a \preceq b$ means $a(x) \leq b(x)$ for all $x\in\cX$.

The estimator $\hFDP_t$ is computed by comparing the number $R_t$ of rejections to the number $A_t$ of $p$-values for which $p_i \geq 1 - s_t(x_i)$:
\[
  R_t = |\{i:\; p_i \leq s_t(x_i)\}|, \qquad
  A_t = |\{i:\; p_i \geq 1-s_t(x_i)\}|, \qquad
  \text{and} \quad \hFDP_t = \frac{1 + A_t}{R_t \vee 1} \;.
\]
The estimate $\hFDP_t$ is also used by \citet{lei2016power} and \citet{arias16}. Figure~\ref{fig:illustrate:compute} illustrates the way $s_t(x)$ and $1-s_t(x)$ partition the data into three regions; $A_t$ is the number of points in the upper blue region and $R_t$ is the number in the lower red region.

\begin{figure}
  \centering
  \begin{subfigure}[t]{.4\textwidth}
    \includegraphics[width=\textwidth]{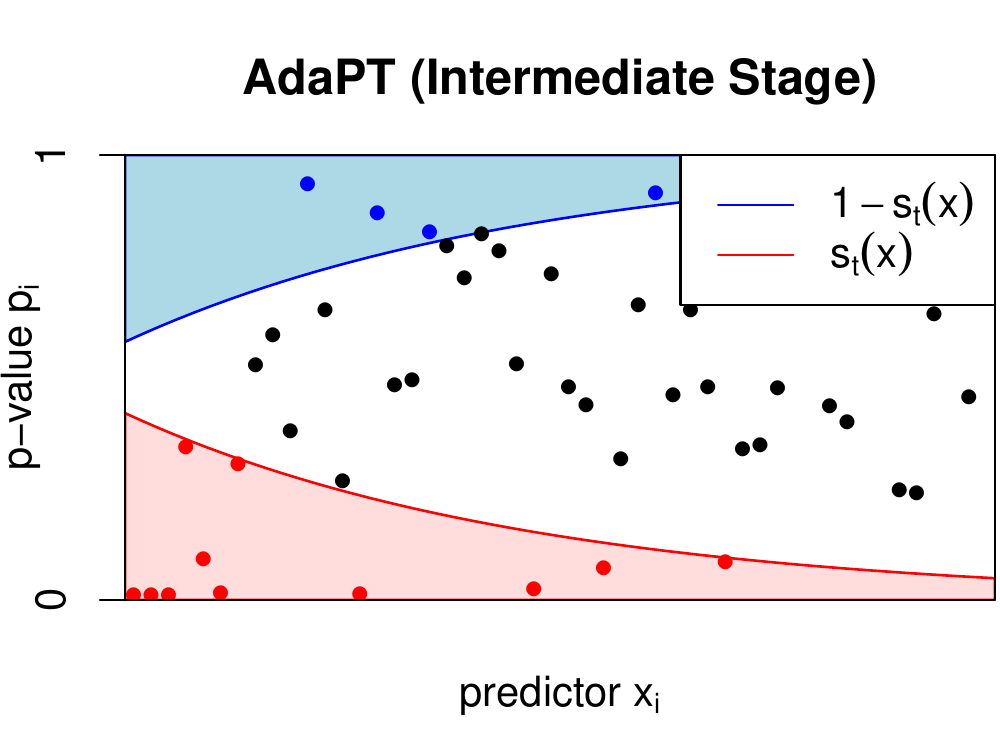}
    \caption{$A_t=4$ and $R_t=11$ are the numbers of blue and red points respectively, leading to $\hFDP = (1+4)/11 \approx 0.45$. If $\hFDP\leq \alpha$, we stop and reject the red points; otherwise we choose a new threshold $s_{t+1} \preceq s_t$ and continue.}
    \label{fig:illustrate:compute}
  \end{subfigure}
  \hspace{.1\textwidth}
  \begin{subfigure}[t]{.4\textwidth}
    \includegraphics[width=\textwidth]{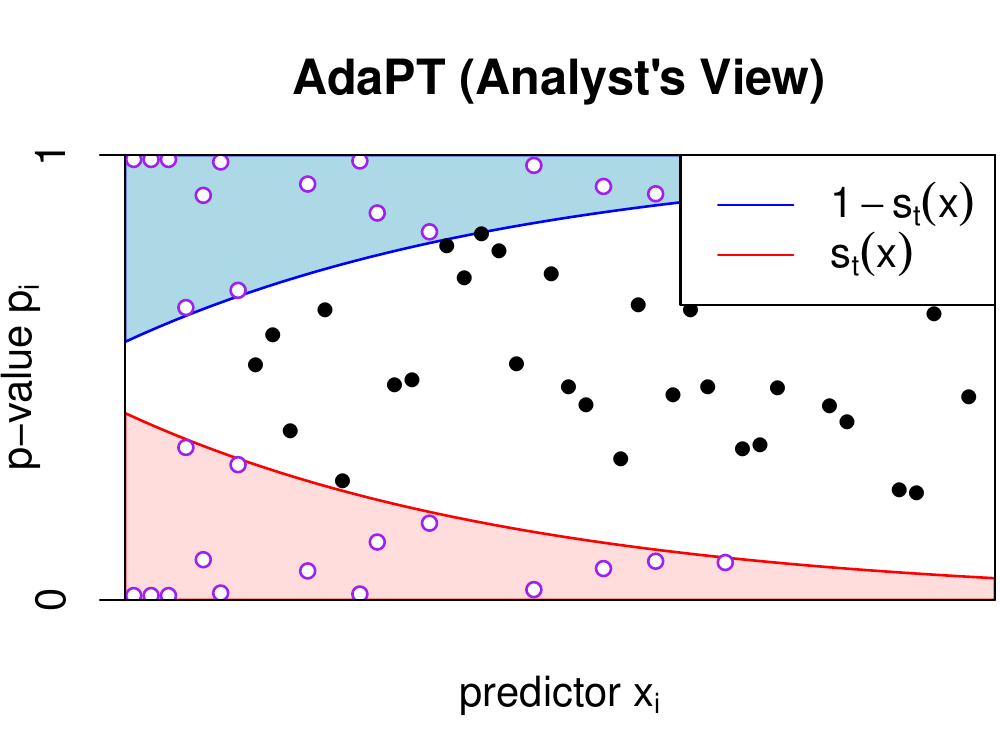}
    \caption{Information available to the analyst when choosing $s_{t+1}(x)$ ($A_t$ and $R_t$ are also known). Each red and blue point is reflected across $p=0.5$, leaving the analyst to impute which are the true $p$-values and which are the mirror images.}
  \label{fig:illustrate:analyst-view}
  \end{subfigure}
  \caption{Illustration of one step of the AdaPT procedure with a univariate predictor.}
  \label{fig:illustrate}
\end{figure} 

At each step $t$, the analyst can choose the next threshold $s_{t+1}(x)$ however she chooses, with only two constraints. First, $s_{t+1}\preceq s_t$ as stated before. Second, the large and small $p$-values (the ones contributing to $A_t$ and $R_t$) are partially masked. Specifically, at step $t$ the analyst is allowed to observe $A_t$ and $R_t$, as well as the entire sequence $(x_i, \tp_{t,i})_{i=1}^n$, where
\begin{equation}\label{eq:mask}
  \tp_{t,i} = 
  \begin{cases} 
    p_i & s_t(x_i) < p_i < 1 - s_t(x_i) \\ 
    \{p_i, \; 1-p_i\} & \text{otherwise}.
  \end{cases}
\end{equation}
Thus, if $p_i = 0.01 \leq s_t(x_i)$ then at step $t$ the analyst knows only that $p_i$ is either 0.01 or 0.99, but if $s_{t+1}(x_i) < 0.01$ then $p_i$ is revealed at step $t+1$ as 0.01. Figure~\ref{fig:illustrate:analyst-view} illustrates what the analyst can see: each red and blue point from Figure~\ref{fig:illustrate:compute} is shown along with its mirror image reflected across the midline $p=0.5$.

We show in Section~\ref{sec:admissible} that, in a generic two-groups empirical Bayes model, an ideal choice for $s_{t}(x)$ would be a level surface of the {\em local false discovery rate} (fdr), as a function of $x$ and $p$:
\[
\fdr(p \mid x) = \P(H_i \text{ is null} \mid p_i = p, x_i = x).
\]
Formally, $\fdr(p\mid x)$ is unidentifiable from the data but, under reasonable assumptions, we can use a good proxy based on the conditional density of the $p$-value given the covariate, $f(p \mid x)$ (note however that our method controls FDR without any empirical Bayes assumptions).

In each step information is gradually revealed to the analyst as the threshold shrinks and more $p$-values are unmasked. Our procedure is adaptive in an unusually strong sense: provided that the two constraints are met, the analyst may apply any method she wants to select $s_{t+1}(x)$, consulting her own hunches or the intuition of domain experts, and can even switch between different methods as information accrues. Moreover, the analyst is under no obligation to describe, or even to fully understand, her update rule for choosing $s_{t+1}(x)$. In this sense, we say our method is fully {\em interactive} --- the analyst's behavior is arbitrary as long as she abides by a certain protocol for interacting with the algorithm. 

While the partial masking of $p$-values obscures just enough information from the analyst to control the FDR, in many cases it does not seriously impact the ability of the analyst to learn the optimal threshold surface $s(x)$. This is because, by the time the algorithm is close to stopping, the vast majority of $p$-values have already been revealed, and many of the ones that remain masked are so minuscule as to leave little doubt about whether $p_i$ is large or small. As we show in numerous simulation and real data experiments in Section~\ref{sec:experiments}, the fdr estimates based on masked data typically converge to the full-data estimates well before the algorithm stops.

The AdaPT procedure controls FDR at level $\alpha$ in finite samples provided that the null $p$-values are uniform, or mirror-conservative as defined in Section \ref{subsec:notation}, and independent conditional on the non-null $p$-values. The proof relies on a pairwise exchangeability argument similar to the argument in \citet{barber2015controlling}.

Algorithm \ref{algo:adapt} summarizes the AdaPT procedure, using the generic sub-routine $\update$ to represent whatever process the analyst uses to select $s_{t+1}(x)$. Note that $s_{t+1}(x)$ is a random function that is measurable to $\cF_{t}$. Sections~\ref{sec:admissible}--\ref{sec:implementation} discuss recommendations for a good $\update$ routine. It is worth mentioning that AdaPT reduces to Barber-Cand{\`e}s method, inspired by \cite{barber2015controlling} and proposed by \cite{arias16}, when $s_{t}(x)$ is a constant function for every $t$.

\begin{algorithm}
  \caption{\adapt}
  \textbf{Input: } predictors and $p$-values $\left(x_{i}, p_{i}\right)_{i \in [n]}$, initialization $s_{0}$, target FDR level $\alpha$
  
  \textbf{Procedure: }
  \begin{algorithmic}[1]
    \For{$t = 0, 1, \ldots$}
    \State $\widehat{\FDP}_{t}\gets \frac{1 + A_{t}}{R_{t}\vee 1}$;
    \If{$\widehat{\FDP}_{t}\le \alpha$}
    \State Reject $\{H_i:\; p_i \leq s_t(x_i)\}$;
    \State Return $s_{t}$;
    \EndIf
    \State $s_{t+1}\gets 
    \update(\left(x_{i}, \tp_{t,i}\right)_{i \in [n]}, A_t, R_t, s_t)$;
    \EndFor
  \end{algorithmic}
  \label{algo:adapt}
\end{algorithm}

\subsection{Related work}
In recent work \citet{ignatiadis2016data} propose a different method {\em independent hypothesis weighting} (IHW) for multiple testing with side information. They first bin the predictors into groups $g_1, \ldots,g_K$, and then apply the weighted-BH procedure at level $\alpha$ with piecewise-constant weights; i.e., if $x_i\in g_k$, then $w_i=w(g_k)$. The weights $w(g_1),\ldots,w(g_K)$ are chosen to maximize the number of rejections. This proposal is similar in spirit to the AdaPT procedure since it attempts to find optimal weights, but it is a bit more limited: first, binning the data may be difficult if the predictor space $\cX$ is multivariate or more complex; and second, their method is only guaranteed to control FDR asymptotically, as the number of bins stays fixed and the number of hypotheses in each bin grows to infinity. As a result, we must trust that $n$ is large enough to support however many bins we have chosen to use. By contrast, AdaPT can use any machine-learning method to estimate $\hat f(p \mid x)$, and we can ``overfit away'' without fear of compromising finite-sample FDR control (though overfitting can of course reduce our power if our fdr estimates are too noisy). Another method is proposed by \citet{du2014single} when the covariate is an auxiliary univariate p-value derived by prior information. However, similar to \citet{ignatiadis2016data}, it only controls FDR asymptotically under the fairly strong conditions that the p-values are symmetrically distributed under the null and bounded by $\frac{1}{2}$ under the alternative.

Perhaps the procedure most closely related to ours is the {\em structure-adaptive BH algorithm} or SABHA \citep{li2016multiple}. SABHA first censors the $p$-values below at a fixed level $\tau$ ($\tau=0.5$ in their simulations), leading to censored $p$-values $p_i\1\{p_i > \tau\}$. Using these, they can estimate $\pi_1(x)$, defined as $P(H_i \mbox { is non-null} \mid x_i = x)$, as a function of $x$, then apply the weighted BH procedure of \citet{genovese2006false} with weights $\hat \pi_1(x_i)^{-1}$, at a corrected FDR level $\tilde\alpha = C\alpha$ (where $C<1$ depends on the Rademacher complexity of the estimator $\hat\pi_1^{-1}$). We notice that this type of censoring is also employed in a variant of IHW \citep{ignatiadis2017covariate}, which guarantees the FDR control in finite samples.

As the first procedure to provably control the finite-sample FDR using generic feature information, SABHA represents a major step forward. However, AdaPT has several important advantages: First, even if $\hat\pi_1(x)$ estimates $\pi_1(x)$ consistently, the weights $\pi_1(x)^{-1}$ are not Bayes optimal as we show in Section~\ref{sec:admissible}; by contrast, our method estimates a Bayes optimal threshold. Second, the correction factor $C$ makes the method conservative and restricts the available estimators $\hat\pi_1^{-1}$ to those with provably low Rademacher complexity. Third, AdaPT can use more information for learning: in later stages we will typically have $s_t(x_i) \ll 0.5$ and the masked $p$-values $\tp_{t,i}$ may be much more informative than $p_i\1\{p_i > 0.5\}$, especially since our goal is to estimate $f(p \mid x)$ for small values of $p$.

Finally, we remark that there is a literature on very different approaches for incorporating covariates into multiple testing problems; see e.g. \cite{lewinger2007hierarchical, ferkingstad2008unsupervised, lawyer2009local, zablocki2014covariate}. Unlike our method (and IHW and SABHA), these approaches hinge on the correct specification of the model and might lose the statistical guarantee if the proposed model deviates from the ground truth. By contrast, our method (and IHW and SABHA) rely only on validity of $p$-values (see assumptions of Theorem \ref{thm:fdr} in next Section) and guarantee FDR control even when employing a misspecified model.

\subsection{Outline}

Section~\ref{sec:adapt-proc} defines the AdaPT procedure more formally and gives our main result: if the null $p$-values are independent and mirror-conservative (defined below), AdaPT controls FDR at level $\alpha$ in finite samples.  Section~\ref{sec:admissible} explains why selection of $s_{t+1}(x)$ will typically operate by first estimating the conditional density $f(p \mid x)$ as a function of $x$, and Section~\ref{sec:implementation} gives practical suggestions for update rules. Section~\ref{sec:experiments} illustrates the AdaPT procedure's power on five real datasets and two simulated datasets, and Section~\ref{sec:discussion} concludes. The programs to replicate all our experiments can be obtained from $\texttt{https://github.com/lihualei71/adaptPaper/}$. Our $\texttt{R}$ package $\texttt{adaptMT}$ can be found in $\texttt{https://github.com/lihualei71/adaptMT/}$.

\section{The AdaPT procedure}\label{sec:adapt-proc}

\subsection{Notation and assumptions}\label{subsec:notation}

Let $[n]$ denote the set $\{1, \ldots, n\}$. For each hypothesis $H_i$, $i\in [n]$ we observe $x_i\in \cX$ and $p_i\in [0,1]$. Let $\cH_0$ denote the set of true null hypotheses. We will assume throughout that $(p_i)_{i\in \cH_0}$ are mutually independent, and independent of $(p_i)_{i\notin \cH_0}$ (see Section~\ref{sec:discussion} for a discussion of how we might relax the independence assumption). Finally, for each $i\in \cH_0$, we assume that $p_i$ is either uniform or mirror-conservative in a sense we will define shortly.

Let $\cF_t$ for $t=0,1,\ldots$ represent the filtration generated by all information available to the user at step $t$:
\[
\cF_t = \sigma\left( (x_i, \tp_{t,i})_{i=1}^n, A_t, R_t \right).
\]
We similarly define an initial $\sigma$-field with all $p$-values masked, $\cF_{-1} = \sigma\left( (x_i, \{p_i,1-p_i\})_{i=1}^n\right).$
The $p$-value masking is equivalent to requiring that $s_{t+1}\in \cF_t$. (For simplicity we have implicitly ruled out the possibility that the analyst uses a randomized rule to update the threshold, but this restriction could be easily removed.) The two constraints $s_{t+1} \preceq s_t$ and $s_{t+1}\in \cF_t$ ensure that $(\cF_t)_{t=-1,0,1,\ldots}$ is a filtration; i.e., the information in $\cF_t$ only grows from $t$ to $t+1$:
\begin{lemma}
  For all $t\geq -1$, $\cF_{t} \subseteq \cF_{t+1}$.
\end{lemma}
\begin{proof}
  We use induction on $u$ to show that $\cF_{u} \subseteq \cF_{t}$ for any $u\leq t$. The conclusion is trivial for $u = -1$ since $\{p_i,1-p_i\}$ is always computable from $p_{t,i}$ (masked $p$-values can always be computed from masked or unmasked ones).

For $u\geq 0$, note that, by the inductive assumption, $s_{u}\in \cF_{u-1} \subseteq \cF_{t}$. As a result, we can compute $p_{u,i}$ which depends only on $p_{t,i}$ and $s_u(x_i)$. Furthermore, 
\[
R_u = R_t + \#\{i:\; p_{t,i} \in (s_t(x_i),s_u(x_i)]\}, \quad 
A_u = A_t + \#\{i:\; p_{t,i} \in [1-s_u(x_i), 1-s_t(x_i))\},
\]
completing the proof.
\end{proof}

To avoid trivialities we assume that the analyst always reveals at least one censored $p$-value in each step of the algorithm, since there is no reason ever to update the threshold surface in a way that reveals no new information. Thus, the stopping time $\hat{t} \leq n$ almost surely.

In many common settings, null $p$-values are conservative but not necessarily exactly uniform. For example, $p$-values from permutation tests are discrete, and $p$-values for composite null hypotheses are often conservative if the true value of the parameter lies in the interior of the null. 

Our method does not require uniformity, but the standard definition of conservatism --- that $\P_{H_i}(p_i \leq a) \leq a$ for all $0 \leq a \leq 1$ --- is {\em not} enough to guarantee FDR control. Instead, we say that a $p$-value $p_i$ is {\em mirror-conservative} if
\begin{equation}\label{eq:mirror-conservative}
\P_{H_i}(p_i \in [a_1,a_2]) \leq \P_{H_i}(p_i \in [1-a_2,1-a_1]), \quad\text{ for all }  0 \leq a_1 \leq a_2 \leq 0.5.
\end{equation}
If $p_i$ is discrete, \eqref{eq:mirror-conservative} means $p_i=1-a$ is at least as likely as $p_i=a$ for $a\leq 0.5$; if $p_i$ has a continuous density, it means the density is at least as large at $1-a$ as at $a$. Mirror-conservatism is not a consequence of conservatism (take $p_i = 0.1 + 0.9 B$ where $B \sim \text{Bernoulli}(0.9)$), and neither does it imply conservatism (take $p_i = B$). Any null distribution with an increasing density is evidently both conservative and mirror-conservative.

Permutation $p$-values are mirror-conservative, as are $p$-values for one-sided tests of univariate parameters with monotone likelihood ratio (with discrete $p$-values randomized to be uniform at the boundary between the null and alternative). See Appendix \ref{subapp:mirror} for proofs of these claims.

\subsection{FDR control}

We are now prepared to prove our main result: the AdaPT procedure controls FDR in finite samples. The proof relies on a similar optional stopping argument as the one presented in \citet{lei2016power} and \citet{barber2016knockoff} (themselves modifications of arguments in \citet{storey04} and \citet{barber2015controlling}). Let $V_t$ and $U_t$ denote the numbers of {\em null} $p_i \leq s_t(x_i)$ and {\em null} $p_i \geq 1-s_t(x_i)$, respectively. If the null $p$-values are uniform then, no matter how we choose $s_t(x)$ at each step, we will always have $V_t \approx U_t$ and $\hFDP_t > \frac{U_t}{R_t \vee 1} \approx \frac{V_t}{R_t \vee 1}$.

\begin{lemma}\label{lem:bernoulli}
  Suppose that, conditionally on the $\sigma$-field $\cG_{-1}$, $b_1,\ldots,b_n$ are independent Bernoulli random variables with $\P(b_i = 1 \mid \cG_{-1}) = \rho_i \geq \rho > 0$, almost surely. Also suppose that $[n] \supseteq \cC_0 \supseteq \cC_1 \supseteq \cdots$, with each subset $\cC_{t+1}$ measurable with respect to
  \[
  \cG_t = \sigma\left(\cG_{-1}, \cC_t, (b_i)_{i \notin \cC_t}, \sum_{i \in \cC_t} b_i\right).
  \]

  If $\htt$ is an almost-surely finite stopping time with respect to the filtration $(\cG_t)_{t \geq 0}$, then
  \[
  \E\left[\frac{1 + |\cC_{\htt}|}{1 + \sum_{i\in \cC_{\htt}} b_i} \mid  \cG_{-1}\right]  \leq \rho^{-1}.
  \]
\end{lemma}

Our Lemma~\ref{lem:bernoulli} generalizes Lemma 1 in \citet{barber2016knockoff} and uses a very similar technical argument. The proof is given in the appendix. Using Lemma~\ref{lem:bernoulli}, we can give our main result:

\begin{theorem}\label{thm:fdr}
  Assume that the null $p$-values are independent of each other and of the non-null $p$-values, and the null $p$-values are uniform or mirror-conservative. Then the AdaPT procedure controls the FDR at level $\alpha$, conditional on $\cF_{-1}$ and also marginally.
\end{theorem}
\begin{proof}
  Let $\htt$ denote the step at which we stop and reject. Then
  \[
    \FDP_{\htt} 
    \;\;=\;\; \frac{V_{\htt}}{R_{\htt} \vee 1} 
    \;\;=\;\;  \frac{1+U_{\htt}}{R_{\htt} \vee 1} \,\cdot\,\frac{V_{\htt}}{1 + U_{\htt}} 
    \;\;\leq\;\; \alpha \frac{V_{\htt}}{1 + U_{\htt}},
  \]
  where the last step follows from the stopping condition that $\hFDP_{\htt}\leq \alpha$, and the fact that $U_t \leq A_t$. We will finish the proof by establishing that $\E[V_{\htt}/(1+U_{\htt})] \leq 1$, using Lemma~\ref{lem:bernoulli}. 

  Let $m_i = \min\{p_i,1-p_i\}$ and $b_i = \1\left\{p_i \geq 0.5\right\}$, so $p_i = b_i (1-m_i) + (1-b_i) m_i$. Then knowing $b_i$ and $m_i$ is equivalent to knowing $p_i$. Let $\cC_t = \{i\in \cH_0:\; p_i \notin (s_t(x_i), 1-s_t(x_i))\}$, representing the null $p$-values that are {\em not} visible to the analyst at time $t$. Then,
  \[
  U_t = \sum_{i\in \cC_t} b_i, \quad \text{ and } \quad V_t = \sum_{i\in \cC_t} (1-b_i) = |\cC_t|-U_t.
  \]

  Further, define the $\sigma$-fields
  \[
  \cG_{-1} = \sigma\left( \left(x_i, m_i\right)_{i=1}^n, \; (b_i)_{i\notin \cH_0} \right), \quad \text{ and } \quad \cG_t = \sigma\left( \cG_{-1}, \cC_t, (b_i)_{i \notin \cC_t}, U_t \right).
  \]

  The assumptions of independence and mirror-conservatism guarantee $\P\left(b_i = 1 \mid \cG_{-1}\right) \geq 0.5$ almost surely for each $i \in \cH_0$, with the $b_i$ conditionally independent.
  
  Next, note that $\cF_t \subseteq \cG_t$ because $p_i\in \cG_t$ for each $p_i \in (s_t(x_i), 1-s_t(x_i))$, and
  \[
  A_t = U_t + \left|\{i\notin \cH_0:\; p_i \geq 1-s_t(x_i)\}\right|,
  \]
  and $R_t \in \cG_t$ by a similar argument. It follows that ${\htt = \min \{t:\; \hFDP_t \leq \alpha\}}$ is a stopping time with respect to $\cG_t$; furthermore, ${\cC_{t+1} \in \cF_t \subseteq \cG_t}$ by assumption.

  As a result, conditional on $\cG_{-1}$, we can apply Lemma~\ref{lem:bernoulli} to obtain
  \[
  \E [\FDP \mid \cG_{-1}] \;\;\leq\;\; \alpha \,\E\left[\frac{V_{\htt}}{1 + U_{\htt}} \mid \cG_{-1}\right]
  \;\;=\;\; \alpha \,\E\left[\frac{1 + |\cC_{\htt}|}{1 + U_{\htt}} - 1\mid \cG_{-1}\right]
  \;\;\leq\;\; \alpha \left( 2 - 1\right)
  \;\;=\;\; \alpha.
  \]
Note that $\cF_{-1}\subset \cG_{-1}$. The proof is completed by applying the tower property of conditional expectation. 
\end{proof}

The main technical point of departure for our method is that the optional stopping argument is not merely a technical device to prove FDR control for a fixed algorithm like the BH, Storey-BH, or Knockoff+ procedures. Instead, we push the optional-stopping argument to its limit, allowing the analyst to interact with the data in a much more flexible and adaptive way. Sections~\ref{sec:hybrid}--\ref{sec:knockoffs} further investigate the connection to knockoffs.

\section{A Guideline To Choose Thresholding Rules}\label{sec:admissible}

Although the AdaPT procedure controls FDR no matter how we update the threshold, its power depends on the quality of the updates. This section concerns the question of what thresholds we would choose if we had perfect knowledge of the data-generating distribution, with Section~\ref{sec:implementation} discussing suggestions for learning optimal thresholds from the data. To establish a guideline for threshold update, we consider a conditional two-groups model as the \emph{working model}. As we will see, under mild conditions, the Bayes-optimal rejection thresholds are the level surfaces of the {\em local false discovery rate} (fdr), defined as the probability that a hypothesis is null conditional on its $p$-value. The local FDR was first discussed by \citet{efron2001empirical}; see also \citet{efron2007size}. A similar result is obtained by \citet{storey2007optimal} under a different framework.

\subsection{The two-groups model and local false discovery rate}

To begin, we assume a {\em two-groups model} conditional on the predictors $x_i$. Letting $H_i=0$ if the $i$th null is true and $H_i=1$ otherwise, we assume:
\begin{align*}
  H_i \mid x_i &\sim \text{Bernoulli}(\pi_1(x_i))\\
  p_i \mid H_i, x_i &\sim \begin{cases} f_0(p \mid x_i) & \text{ if } H_i = 0\\ f_1(p \mid x_i) & \text{ if } H_i = 1\end{cases}.
\end{align*}
In addition, we assume that $(x_i, H_i, p_i)$ are independent for $i\in [n]$. Unless otherwise stated we will assume for simplicity that both $f_0$ and $f_1$ are continuous densities, with $f_0(p \mid x) \equiv 1$ (null $p$-values are uniform) and $f_1(p \mid x)$ non-increasing in $p$ (smaller $p$-values imply stronger evidence against the null). Furthermore, define the conditional mixture density
\[
f(p \mid x) \;=\; \left(1-\pi_1(x)\right) f_0(p \mid x) + \pi_1(x) f_1(p \mid x) \;=\; 1-\pi_1(x) + \pi_1(x) f_1(p \mid x),
\]
and the conditional local false discovery rate
\[
\fdr(p \mid x) \;=\; \P(H_i \text{ is null } \mid x_i=x, p_i=p) 
\;=\; \frac{1 - \pi_1(x)}{f(p \mid x)}.
\]

Note that we never observe $H_i$ directly. Thus, while $f$ is identifiable from the data, $\pi_1$ and $f_1$ are not: for example, $\pi_1 = 0.5, f_1(p\mid x) = 2(1 - p)$ and $\pi_1 = 1, f_1(p\mid x) = 1.5 - p$ result in exactly the same mixture density. Unless $f_1(p \mid x)$ is known {\em a priori}, we can make the conservative identifying assumption that
\[
1-\pi_1(x) = \inf_{p \in [0,1]} f(p \mid x) = f(1 \mid x),
\]
attributing as many observations as possible to the null hypothesis. This approximation is very good when $\fdr(1 \mid x) \approx 1$, which is reasonable in many settings. Thus, any estimate $\hat{f}$ of the mixture density translates to a conservative estimate $\hfdr(p \mid x) = \hat{f}(1 \mid x) / \hat{f}(p \mid x)$.

\subsection{Optimal thresholds under the two-groups model}
Let $\nu$ be a probability measure on $\mathcal{X}$ and define a random variable $X\sim \nu$. Similar to \citet{sun15}, for any thresholding rule $s(x)$, we define the global FDR as
\[
\mathrm{FDR}(s; \nu) = \P(H = 0 \mid H \mbox{ is rejected}) = \P(H = 0 \mid P\le s(X))
\]
where $H$ and $P$ are a hypothesis and $p$-value distributed according to the two-groups model. The power is defined in a similar fashion as
\[
\mathrm{Pow}(s; \nu) = \P(H \mbox{ is rejected}\mid H = 1) = \P( P \le s(X) \mid H = 1).
\]
\citet{sun15} formulates a compound decision-theoretic framework by defining a Bayesian-type loss function. Instead, we propose a Neyman-Pearson type framework, i.e. 
\begin{equation}\label{eq:np}
\max_{s}\; \mathrm{Pow}(s; \nu)\quad \mathrm{s.t.} \,\,\mathrm{FDR}(s; \nu)\le \alpha.
\end{equation}

Next, define
\begin{align*}
  Q_0(s) &= \P(P \le s(X), H = 0) = \int_{\mathcal{X}}F_{0}(s(x) | x)(1-\pi_1(x))\nu(dx)\\
  Q_1(s) &= \P(P \le s(X), H = 1) = \int_{\mathcal{X}}F_{1}(s(x) | x)\pi_1(x)\nu(dx),
\end{align*}
where $F_0$ and $F_1$ are the cumulative distribution functions under the null and alternative. We can simplify~\eqref{eq:np} as
\begin{align}
  \max_{s}\; &\frac{Q_1(s)}{\P(H=1)} \quad \mathrm{s.t.} \,\,\frac{Q_0(s)}{Q_0(s) + Q_1(s)}\le \alpha\\
  \iff \min_{s}\; &-Q_1(s) \qquad \mathrm{s.t.} \,\, -\alpha Q_1(s) + (1-\alpha) Q_0(s) \le 0\\
  \iff \min_{s}\; &\int_{\mathcal{X}}-F_{1}(s(x) | x)\pi_1(x)\nu(dx)\nonumber\\
  \quad \mathrm{s.t.}& \,\, \int_{\mathcal{X}}\bigg\{-\alpha F_{1}(s(x) | x)\pi_1(x) + (1 - \alpha)F_{0}(s(x) | x)(1 - \pi_1(x))\bigg\}\nu(dx)\le 0.
                       \label{eq:np_simplify}
\end{align}

The corresponding Lagrangian function can be written as 
\begin{equation}\label{eq:lagrangian}
L(s; \lambda) = \int_{\mathcal{X}}\bigg\{-(1 + \lambda \alpha)F_{1}(s(x) | x)\pi_1(x) + \lambda(1 - \alpha)F_{0}(s(x) | x)(1 - \pi_1(x))\bigg\}\nu(dx).
\end{equation}
Let $s^{*}$ be the optimum, then the Karush-Kuhn-Tucker (KKT) condition (under regularity conditions) implies that 
\begin{align}
& (1 + \lambda \alpha)f_{1}(s^{*}(x) | x)\pi_1(x) = \lambda(1 - \alpha)f_{0}(s^{*}(x) | x)(1 - \pi_1(x))\nonumber\\
& \Longrightarrow \mathrm{fdr}(s^{*}(x) | x) = \frac{1 + \lambda \alpha}{1 + \lambda}.\label{eq:level_surface}
\end{align}
In other words, the optimal thresholding rules are level surfaces of local FDR. Theorem~\ref{thm:level_surface} formalizes the above derivation by clarifying the regularity conditions.
\begin{theorem}\label{thm:level_surface}
  Assume that 
  \begin{enumerate}[(a)]
  \item $f_{1}(p \mid x_{i})$ is continuously non-increasing and $f_{0}(p \mid x_{i})$ is continuously non-decreasing and uniformly bounded away from $\infty$;
  \item $\nu$ is a discrete measure supported on $\{x_{1}, \ldots, x_{n}\}$ with $\nu(\{x_{i}: \fdr(0 \mid x_{i}) < \alpha, f(0 \mid x_{i}) > 0\}) > 0$.
  \end{enumerate}
 Then \eqref{eq:np} has at least a solution, and all solutions are level surfaces of $\mathrm{fdr}(p \mid x)$.
\end{theorem}

In practice, any conservative null distribution (stochastically dominated by $U([0, 1])$) with positive density at zero satisfies condition (a). The monotonicity of $f_{1}$ is also valid since smaller p-values imply stronger evidence against null. In condition (b), the assumption on the support is reasonable since we treat $\{x_{i}: i \in [n]\}$ as fixed and hence only the quantities associated with these values are of interest. We believe it can be relaxed to more general measures and will not discuss it due to the technical complication. In contrast, the second requirement is necessary since it implies the feasibility of the problem. If the local FDR is above $\alpha$ almost everywhere, no thresholding rule is able to control FDR at $\alpha$. As mentioned above, we can set $s$ as the level surfaces of $\widehat{\mathrm{\fdr}}(p \mid x) = \hat{f}(1 \mid x) / \hat{f}(p \mid x)$ given some estimator $\hat{f}(p\mid x)$. The next section discusses estimation of $\hat{f}(p\mid x)$.

\section{Implementation}\label{sec:implementation}

Having shown that level surfaces of the local FDR are optimal under the two-groups model, we now turn to estimation of $\fdr(p\mid x)$, which boils down to estimation of the conditional density $f(p\mid x)$. This section discusses a flexible framework for conditional density estimation that can perform favorably when no domain-specific expertise can be brought to bear. 

More generally, we should model the data using as much domain-specific expertise as possible. We emphasize once more that, no matter how misspecified our model is, no matter how misguided our priors are (if we use a Bayesian method), no matter how we select a model or tuning parameter, or how much that selection biases our resulting estimate of local FDR, the AdaPT procedure nevertheless controls global FDR. Thus, there is every reason to be relatively aggressive in choosing a modeling strategy. 

\subsection{Conditional density estimation via the expecation maximization algorithm}

Generically, we can model the conditional density by a parametric family where we assume null p-values are uniform distributed, i.e. $f_{0}(p \mid x_i)\equiv 1$, and each non-null p-value has a density in the following exponential family, indexed by a univariate parameter $\eta_i$:
\begin{equation}\label{eq:exp_family1}
  f_{1}(p \mid x_i) = h(p; \eta_i) \triangleq e^{\eta_i g(p) - B(\eta_i)}.
\end{equation}
Note that $\eta_i$ and $g(p)$ can be vectors but we focus on the scalar case for simplicity. Let 
\begin{equation}\label{eq:exp_family_transform}
y_i = g(p_i), \quad \mu_i = B'(\eta_i).
\end{equation}
Using the standard argument, \eqref{eq:exp_family1} implies that 
\begin{align}
&  \E_{\eta_i}[y_i] = \E_{\eta_i}[g(p_i)] = B'(\eta_i) = \mu_i,\label{eq:mu}
\end{align}
where $\E_{\eta_i}$ denotes the expectation under $h(\cdot; \eta_i)$. If $g$ is not almost-everywhere constant, then $B''(\eta)=\text{Var}_\eta(y_i)>0$ and $B'$ is bijective. Then there is a one-to-one mapping from $\mu_i$ to $\eta_i$, denoted by $\eta_i = \eta(\mu_i)$ as convention. In fact, $\eta(\cdot) = (B')^{-1}(\cdot)$. Then \eqref{eq:exp_family1} can be reparametrized using $\mu_i$,
\begin{equation}\label{eq:exp_family2}
  h(p; \mu_i) = e^{\eta(\mu_i) g(p) - A(\mu_i)},
\end{equation}
where $A(\cdot) = B(\eta(\cdot))$ and we abuse the notation $h(p; \cdot)$. As we will see, it is more convenient to use the mean parametrization \eqref{eq:exp_family2}. 

Given \eqref{eq:exp_family2}, it is left to model $\pi_{1i} \triangleq \pi_1(x_i)$ and $\mu_i$ (or $\eta_i$ equivalently). In this article we consider the following generalized linear model where $\phi_{\pi}(x), \phi_{\mu}(x)$ denote two featurization and $\zeta$ denotes a link function:
\begin{align}
  H_i \mid x_i 
  &\sim \text{Bernoulli}(\pi_{1i}), 
    \qquad \text{with } \log\frac{\pi_{1i}}{1-\pi_{1i}} = \theta'\phi_{\pi}(x_i),
    \text{ and}\nonumber\\
  p_i \mid x_i, H_i  &\sim \left\{\begin{array}{ll} 
h(p; \mu_i) & \text{if } H_i=1\\
1 & \text{if } H_i=0 \end{array}\right.,
\quad\text{with } \zeta(\mu_i) = \beta'\phi_{\mu}(x_i).\label{eq:twogroup_GLM}
\end{align}
In particular, $\zeta(\cdot) = \eta(\cdot)$ gives the canonical link function. For instance, when $g(p) = -\log p, \eta(\mu) = -\frac{1}{\mu} + 1$ and $A(\mu) = \log \mu$, 
\begin{equation}\label{eq:beta_mixture}
f(p | x) = \pi_{1i}h(p; \mu_i) + (1-\pi_{1i}) = \pi_{1i}\cdot \frac{1}{\mu_i}p^{\frac{1}{\mu_i} - 1} + (1 - \pi_{1i}).
\end{equation}
This yields a beta-mixture model on the conditional density, which has been considered in literature, e.g. \cite{parker1988identifying, allison2002mixture, pounds2003estimating, markitsis2010censored}.

The fully-observed log-likelihood for the model~\eqref{eq:twogroup_GLM} is 
\begin{align}
  \ell(\theta, \beta; p, H, x) 
  &= \sum_{i=1}^n \left\{H_i \theta ' \phi_{\pi}(x_i) - \log\left(1+e^{-\theta'\phi_{\pi}(x_i)}\right)\right\}\nonumber\\
&  + \sum_{i=1}^{n} H_i\{y_i\cdot \eta\circ \zeta^{-1}(\beta'\phi_{\mu}(x_i)) - A\circ \zeta^{-1}(\beta'\phi_{\mu}(x_i))\}
\end{align}
Because some values of $y_i$ and all values of $H_i$ are unknown, we can use the expectation maximization (EM) algorithm to maximize the partially observed log-likelihood. To simplify estimation, we will proceed as though $A_t$ and $R_t$ are missing, so that the $(y_i, H_i)$ pairs are mutually independent given the predictors. That is, at step $t$ of the AdaPT procedure we attempt to maximize the likelihood of the data $D_t = (x_i, \tp_{t,i})_{i\in [n]}$ and treating $s_t$ as fixed.

Recall that $b_i = I(p_i\ge 0.5)$. There are four possible values of $(b_i, H_i)$, with each pair conditionally independent given $D_t$, and whose probabilities can be efficiently computed for any values of $\theta$ and $\beta$. Let ${r=0,1,\ldots}$ index stages of the EM algorithm (recall $t$ is fixed for the duration of the EM algorithm). For the E-step we compute the expectation of the log-likelihood,
\[
\E_{\hat\theta^{(r - 1)}, \hat\beta^{(r - 1)}}\left[\ell(\theta, \beta; y, H, x) | D_t \right],
\]
which amounts to computing the following quantities:
\begin{align}
\widehat{H}_i^{(r)} &= \E_{\hat\theta^{(r-1)}, \hat\beta^{(r-1)}}[H_i \mid D_t], \quad \text{ and }\label{eq:Hhat}\\
\hat{y}_i^{(r,1)} &= \E_{\hat\theta^{(r-1)}, \hat\beta^{(r-1)}}[y_i H_i \mid D_t, \;H_i=1] / \widehat{H}_i^{(r)},\label{eq:phat}
\end{align}
where $\hat\theta^{(r)}$ and $\hat\beta^{(r)}$ denote the current coefficient estimates. We derive the exact formula for \eqref{eq:Hhat} and \eqref{eq:phat} in Appendix \ref{subapp:Estep}. For the M-step, we set
\begin{align}
  \hat\theta^{(r)}, \hat\beta^{(r)} = \arg\max_{\beta,\theta}\;\; &\E_{\hat\theta^{(r-1)}, \hat\beta^{(r-1)}}\left[\ell(\theta, \beta; y, H, x) \mid D_t\right]\nonumber\\
  = \arg\max_{\beta,\theta}\;\; &\sum_{i=1}^n \widehat{H}_i^{(r)} \theta ' \phi_{\pi}(x_i) - \log\left(1+e^{-\theta'\phi_{\pi}(x_i)}\right)\nonumber\\
   + &\sum_{i=1}^n  \widehat{H}_i^{(r)}\cdot\lb\hat{y}_i^{(r,1)}\cdot \eta\circ \zeta^{-1}(\beta'\phi_{\mu}(x_i)) - A\circ \zeta^{-1}(\beta'\phi_{\mu}(x_i))\rb.\label{eq:Mstep}
\end{align}

The optimization above splits into two separate optimization problems, a logistic regression with predictors $\phi_{\pi}(x_i)$ and fractional responses $\widehat{H}_i^{(r)}$, and a GLM with predictors $\phi_{\mu}(x_i)$, responses $\hat{y}_i^{(r,1)}$, and weights $\widehat{H}_i^{(r)}$. Each of these GLM problems can be solved efficiently using the \texttt{glm} function in \texttt{R} (e.g. \cite{dobson2008introduction}). For $r=0$, we can initialize $\hat\theta^{(0)}$ and $\hat\beta^{(0)}$ by a simple method with details discussed in Appendix \ref{subapp:EM_init}. Algorithm~\ref{algo:em} formalizes the EM algorithm using R pseudocode. The family argument for estimating $\hat{\beta}^{(r)}$ depends on the form of exponential family \eqref{eq:exp_family2}. For example, \eqref{eq:Mstep} yields a Gamma GLM in the beta-mixture model \eqref{eq:beta_mixture}.

\begin{algorithm}
\caption{EM algorithm to estimate $\pi_1(\cdot)$ and $\mu(\cdot)$ based on $D_t = (x_i, \tp_{t,i})_{i\in [n]}$}\label{algo:em}
\textbf{Input:} data $D_t$, number of iterations $m$, initialization $\hat{\theta}^{(0)},\hat{\beta}^{(0)}$;
\begin{algorithmic}
  \For{$r = 1, 2,\ldots, m$}
  \State (\emph{E-step}):
  \State \qquad $\widehat{H}_i^{(r)} \gets \E_{\hat\theta^{(r-1)}, \hat\beta^{(r-1)}}[H_i \mid D_t], \quad i\in [n]$;
  \State \qquad $\hat{y}_i^{(r,1)} \gets \E_{\hat\theta^{(r-1)}, \hat\beta^{(r-1)}}[y_i H_i \mid D_t, H_i = 1] / \widehat{H}_{i}^{(r)}, \quad i\in [n]$;
  \State (\emph{M-step}): 
  \State \qquad $\hat\theta^{(r)} \gets \texttt{glm}\left(\widehat{H}^{(r)} \sim \phi_{\pi}(x), \texttt{ family = binomial}\right)$;
  \State \qquad $\hat\beta^{(r)} \gets \texttt{glm}\left(\hat{y}^{(r,1)} \sim \phi_{\mu}(x), \texttt{ family = ...} (\texttt{link} =  \zeta), \,\,\texttt{weights = } \widehat{H}^{(r)}\right)$;
  \EndFor
\end{algorithmic}
\textbf{Output: } $\hat{\pi}_1(x) = \left(1+e^{-\phi_{\pi}(x)'\hat\theta^{(m)}}\right)^{-1}, \;\;\hat{\mu}(x) = \zeta^{-1}\lb\phi_{\mu}(x)'\hat\beta^{(m)}\rb$.
\end{algorithm}

The GLM model~\eqref{eq:twogroup_GLM} provides the starting point for an extremely flexible and extensible modeling framework. More generally, we could replace the fitting procedure in M-step by penalized GLM (\texttt{glmnet} package), generalized additive model (\texttt{gam} or \texttt{mgcv} package), or generalized boosting regression (\texttt{gbm} package). Furthermore, noting that 
\[\pi_1(x) = \E[H \mid x], \quad\mu(x) = \E [y \mid x, H=1],\]
 one can even fit them directly using any nonparametric method, such as random forest or neural networks, that targets on estimating conditional mean. 

\subsection{Selecting featurization}\label{subsec:feature}
Suppose we are given a finite set of candidate featurization $\{(\phi_{\pi, j}(x), \phi_{\mu, j}(x)): j = 1, \ldots, M\}$. For instance for univariate $x$, $\phi_{\pi, j}(x)$ and $\phi_{\mu, j}(x)$ could be spline bases with certain numbers of equi-spaced knots; for multivariate $x$, $\phi_{\pi, j}(x)$ and $\phi_{\mu, j}(x)$ could be subsets of covariates contained in $x$. At step $t$, one is permitted to fit a model for each featurization, using arbitrary methods (e.g., GLM, penalized GLM, etc.), based on $((\phi_{\pi, j}(x_i), \phi_{\mu, j}(x_i), \td{p}_{t,i})_{i=1}^{n})$. Let $\hat{\pi}_{1}^{(j)} = (\hat{\pi}_{11}^{(j)}, \ldots, \hat{\pi}_{1n}^{(j)})$ and $\hat{\mu}^{(j)} = (\hat{\mu}_{1}^{(j)}, \ldots, \hat{\mu}_{n}^{(j)})$ denote the resulting fitted values. The full log-likelihood, assuming $H_i$ is known, for the GLM model~\eqref{eq:twogroup_GLM} based on $(\phi_{\pi, j}(x), \phi_{\mu, j}(x))$ can be written as
\[\ell_{j}(\pi_{1}, \mu) = \sum_{i=1}^{n}(H_i\log \pi_{1i}^{(j)} + (1 - H_i)\log (1 - \pi_{1i}^{(j)})) + \sum_{i=1}^{n}H_{i}\log h(p_i; \mu_i^{(j)}).\]
Though $\ell_{j}$ is not computable, we can replace it by 
\[\td{\ell}_{j}\triangleq \E_{\hat{\pi}_{1}^{(j)}, \hat{\mu}^{(j)}} [\ell_{j}(\pi_{1}, \mu)].\]
This is precisely the objective of M-step and hence is directly computed from the EM algorithm. 

Based on $\{\td{\ell}_{j}\}_{j=1}^{M}$, we can use any information criterion for featurization selection. Our implementation uses BIC as default, defined as
\[\mathrm{BIC}_{j} = \log n\cdot (\mathrm{df}_{\pi, j} + \mathrm{df}_{\mu, j}) - 2\td{\ell}_{j}\]
where $\mathrm{df}_{\pi, j}$ (resp. $\mathrm{df}_{\mu, j}$) is the degree of freedom of $\phi_{\pi, j}$ (resp. $\phi_{\mu, j}$). For instance, $\mathrm{df}_{\pi, j}$ is the number of knots plus 1 (for the intercept) when $\phi_{\pi, j}$ is the spline basis; $\mathrm{df}_{\pi, j}$ is the number of selected covariates plus 1 (for the intercept) when $\phi_{\pi, j}$ is a sparse subset of $x$.

Alternatively, the user can also apply cross-validation to select the featurization. Specifically, at step $t$ the data is divided into $K$ folds. For $k$-th fold, the expected log-likelihood $\td{\ell}_{jk}$ is computed by taking the $k$-th fold as the holdout set and fitting the parameters on other folds. The selection is then based on $\td{\ell}_{j} = \sum_{k=1}^{K}\td{\ell}_{jk}$.

We emphasize that any of above selection procedures can be performed in any intermediate step of AdaPT. If the featurization selection can be computed efficiently, we suggest applying it in every step. Otherwise we suggest performing it only at the first step, in which $s(x) = s_{0}(x)$, and keeping the selected featurization for all later steps.

\subsection{Updating the threshold}

Theorem \ref{thm:level_surface} suggests that our updated threshold $s_{t+1}$ should approximate a level surface of $\hfdr(p \mid x)$. For the model \eqref{eq:twogroup_GLM}, level surfaces of the local FDR are given by
\begin{equation}\label{eq:lfdr_level}
c = \frac{f(1 | x)}{f(s(x) | x)} = \frac{\pi_1(x) h(1; \mu(x))  + 1 - \pi_1(x)}{\pi_1(x)  h(s(x); \mu(x)) + 1 - \pi_1(x)}.
\end{equation}
For various widely-used exponential families in the form \eqref{eq:exp_family2}, $h(p; \mu)$ is decreasing with respect to $p$, in which case,
\begin{equation}\label{eq:sxc}
s(x; c) = f^{-1}\lb \frac{h(1; \mu(x))}{c} + \frac{1 - \pi_1(x)}{\pi_1(x)}\frac{1 - c}{c}; \mu(x)\rb
\end{equation}

Given a chosen local FDR level $c$, we can evolve $s_{t}$ by 
\begin{equation}\label{eq:evolve_first}
s_{t + 1}(x) = \min\{s_{t}(x), s(x; c)\},
\end{equation}
where the minimum is taken to meet the requirement that $s_{t + 1}(x) \le s_{t}(x)$. Note that a higher level surface (larger $c$) will typically give a higher $\hFDP_t$ and vice versa. Unless computational efficiency is at a premium, it is better to force the procedure to be patient since more information can be gained after each update and the learning step can be more accurate. In other words, we shall choose a large $c$ such that $s_{t + 1}(x)$ only deviates from $s_{t}(x)$ slightly. 

In this article we propose a simple procedure to achieve this: it chooses $c$ such that exactly one partially-masked p-value is revealed based on $s_{t+1}(x)$ defined in \eqref{eq:evolve_first}. The choice of $c$ can be computed in the following way
\begin{enumerate}[(a)]
\item Estimate local FDR for each  $p'_{t, i}$ as 
\begin{equation}\label{eq:lfdr}
\mathrm{fdr}_{t,i} = \frac{f(1 | x_i)}{f(p'_{t,i} | x_i)} = \frac{\hat{\pi}_{1i} \cdot h(1; \hat{\mu}_{i})  + 1 - \hat{\pi}_{1i}}{\hat{\pi}_{1i} \cdot h(p'_{t,i}; \hat{\mu}_{i}) + 1 - \hat{\pi}_{1i}},
\end{equation}
where $p'_{t, i}$ is the minimum element in $\td{p}_{t, i}$ (i.e., $\td{p}_{t, i} = p'_{t, i}$ for revealed p-values and $\td{p}_{t, i} = \{p'_{t, i}, 1 - p'_{t, i}\}$ for masked p-values.)
\item Set $c$ as the largest value of $\mathrm{lfdr}_{t,i}$ among all \emph{partially masked p-values}. (Strictly speaking, $c$ should be slightly smaller than $\max_{i}\mathrm{lfdr}_{t,i}$. In implementation we subtract $10^{-15}$ from it.)
\end{enumerate}
As a consequence, this choice of $c$ is measurable with respect to $\cF_t$ and hence a permissible operation in \adapt.

\subsection{Other Issues}

\noindent \textbf{Initial thresholds.} As shown in Algorithm \ref{algo:adapt}, AdaPT starts from some curve $s_{0}(x)$ and then slowly update it. If the hypotheses are not ordered, then we can simply set $s_{0}(x)\equiv s_{0, 1}$ with $s_{0, 1} \le 0.5$. A larger $s_{0, 1}$ is conceptually preferred since the procedure is more patient. We found that $s_{0, 1} = 0.45$ is a consistently good choice.

~\\
\noindent \textbf{Computation efficiency.} The model update (Algorithm \ref{algo:em}) is the most computationally costly component. To save computation, we recommend not updating the model at every step. In our implementation, the default is to update the model every $\lceil n / 20\rceil$ steps.

~\\
\noindent \textbf{$q$-Values.} Rather than specify $\alpha$ in advance, some researchers might prefer to see a list of discoveries for each of a range of $\alpha$ values. Rather than return a single list for a single $\alpha$, we can alternatively run the algorithm once and output {\em $q$-values} for every hypothesis \citep{storey02, storey2003statistical}, defined as the minimum value of $\alpha$ for which the hypothesis would be rejected.

Let $\htt_\alpha = \min\{t:\; \hFDP_t \leq \alpha\}$ and 
\[
t_i^{\ast} = \min\{t:\; s_t(x_i) < p_i < 1-s_t(x_i)\},
\]
the time at which $p_i$ is revealed.  We then see that
\begin{align*}
H_i \text{ rejected at level } \alpha
  &\iff p_i \leq s_{\htt_{\alpha}}(x_i) \\
  &\iff \htt_\alpha < t_i^{\ast} \\
  &\iff \min_{t<t_i^{\ast}} \hFDP_t \leq \alpha
\end{align*}
As a result, $q_i = \min_{t<t_i^{\ast}} \hFDP_t$ is a valid $q$-value for hypothesis $i$.

\section{Experiments}\label{sec:experiments}

\subsection{Gene/Drug response data: an illustrating example}\label{subsec:GEOquery}
To illustrate the power of the AdaPT procedure, we apply it to the GEOquery gene-dosage data~\citep{davis2007geoquery}, which has been analyzed repeatedly as a benchmark for ordered testing procedures \cite{li2016accumulation, lei2016power, li2016multiple}. We use Algorithm \ref{algo:em} with a beta-mixture model \eqref{eq:beta_mixture} for the E-step (see Appendix \ref{subsubapp:GammaGLM} for details) and a Gamma GLM with canonical link function for the M-step. This dataset consists of gene expression measurements for $n = 22283$ genes, in response to estrogen treatments in breast cancer cells for five groups of patients, with different dosage levels and 5 trials in each. The task is to identify the genes responding to a low dosage. The p-values $p_{i}$ for gene $i$ is obtained by a one-sided permutation test which evaluates evidence for a change in gene expression level between the control group (placebo) and the low-dose group. $\{p_{i}: i \in [n]\}$ are then ordered according to permutation $t$-statistics comparing the control and low-dose data, pooled, against data from a higher dosage (with genes that appear to have a strong response at higher dosages placed earlier in the list). 

We consider two orderings: first, a stronger (more informative) ordering based on a comparison to the highest dosage; and second, a weaker (less informative) ordering based on a comparison to a medium dosage. Let $\sigma_S(i)$ and $\sigma_W(i)$ denote respectively the permutations of $i\in [n]$ given by the stronger and weaker orderings. Further details on these two orderings can be found in \citet{li2016accumulation} and \citet{li2016multiple}. We write the $p$-values, thus reordered, as $p_i^S=p_{\sigma_S(i)}$ and $p_i^W =p_{\sigma_W(i)}$. Once the data are reordered, we can apply either a method that ignores the ordering altogether, or an ordered testing procedure, or a testing procedure that uses generic side information, using the index of the reordered $p$-values as a univariate predictor. 

We compare AdaPT against twelve other methods : 
\begin{enumerate}[(a)]
\item SeqStep with parameter $C = 2$ \citep{barber2015controlling};
\item ForwardStop \citep{gsell2016sequential};
\item the accumulation test with the HingeExp function and parameter $C = 2$ \citep{li2016accumulation};
\item Adaptive SeqStep with $s = q$ and $\lambda = 1 - q$ \citep{lei2016power};
\item BH procedure \citep{bh95};
\item Storey's BH procedure with threshold $\lambda = 0.5$ \citep{storey04};
\item Barber-Cand{\`e}s method \citep{barber2015controlling, arias16};
\item SABHA with $\tau = 0.5, \eps = 0.1$ and the stepwise constant weights, monotone taking values in $\{\eps, 1\}$ (see section 4.1 of \citet{li2016multiple});
\item SABHA with $\tau = 0.5, \eps = 0.1$ and the monotone weights, taking values in $[\eps, 1]$ (see section 4.1 of \citet{li2016multiple});
\item Independent Hypothesis Weighting (IHW) with number of bins and folds set as default \citep{ignatiadis2016data};
\item an oracle version of IHW with the number of bins determined by maximizing the number of rejections;
\item an oracle version of Independent Filtering (IF) with the cutoff determined by maximizing the number of rejections \citep{bourgon2010independent}.
\end{enumerate}
 Note that the last two methods do not guarantee FDR control because the optimal parameter is selected; and both versions of SABHA control FDR at level $1.134\alpha$ (Lemma 1 of \cite{li2016multiple}) when the target level is $\alpha$. Despite the potential anti-conservativeness of these methods, we do not make correction in order to compare their best possible performance to \adapt. Figure \ref{fig:GEOquery_rejs} shows the number of discoveries with different target FDR levels. We only show the range of $\alpha$ from $0.01$ to $0.3$ since it is rare to allow FDR to be above 0.3 in practice. We use different featurization for estimating $\pi(x)$ and $\mu(x)$, selected from the combination of all spline basis with $6-15$ equi-quantile knots via BIC criterion at the initial step and kept the same afterwards; see Section \ref{subsec:feature}.

\begin{figure}[H]
  \centering
  \includegraphics[width = 0.85\textwidth]{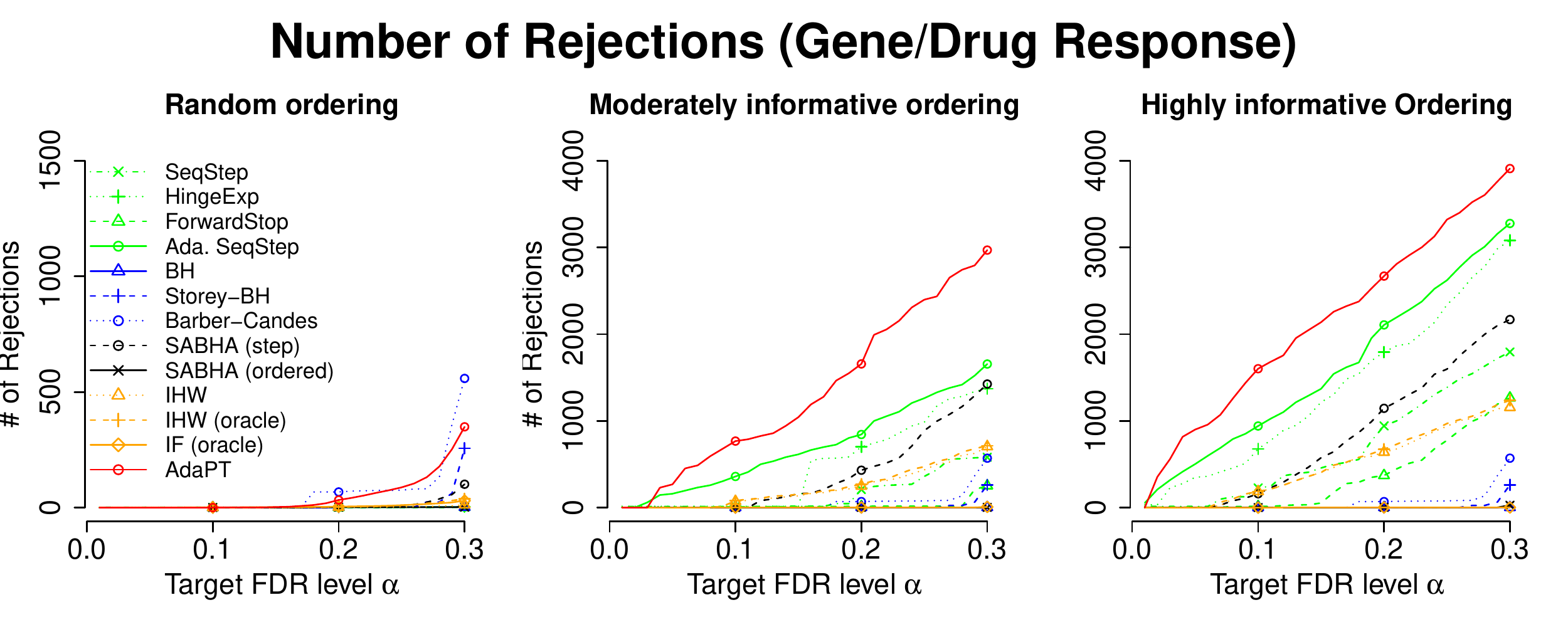}
  \caption{Number of discoveries, in gene/drug response dataset, by each method at a range of target FDR levels $\alpha$ from 0.01 to 0.30. Each panel plots the results for an ordering, ranging from random ordering to highly informative.}
  \label{fig:GEOquery_rejs}
\end{figure}

\begin{figure}[H]
  \centering
  \includegraphics[width = 0.49\textwidth]{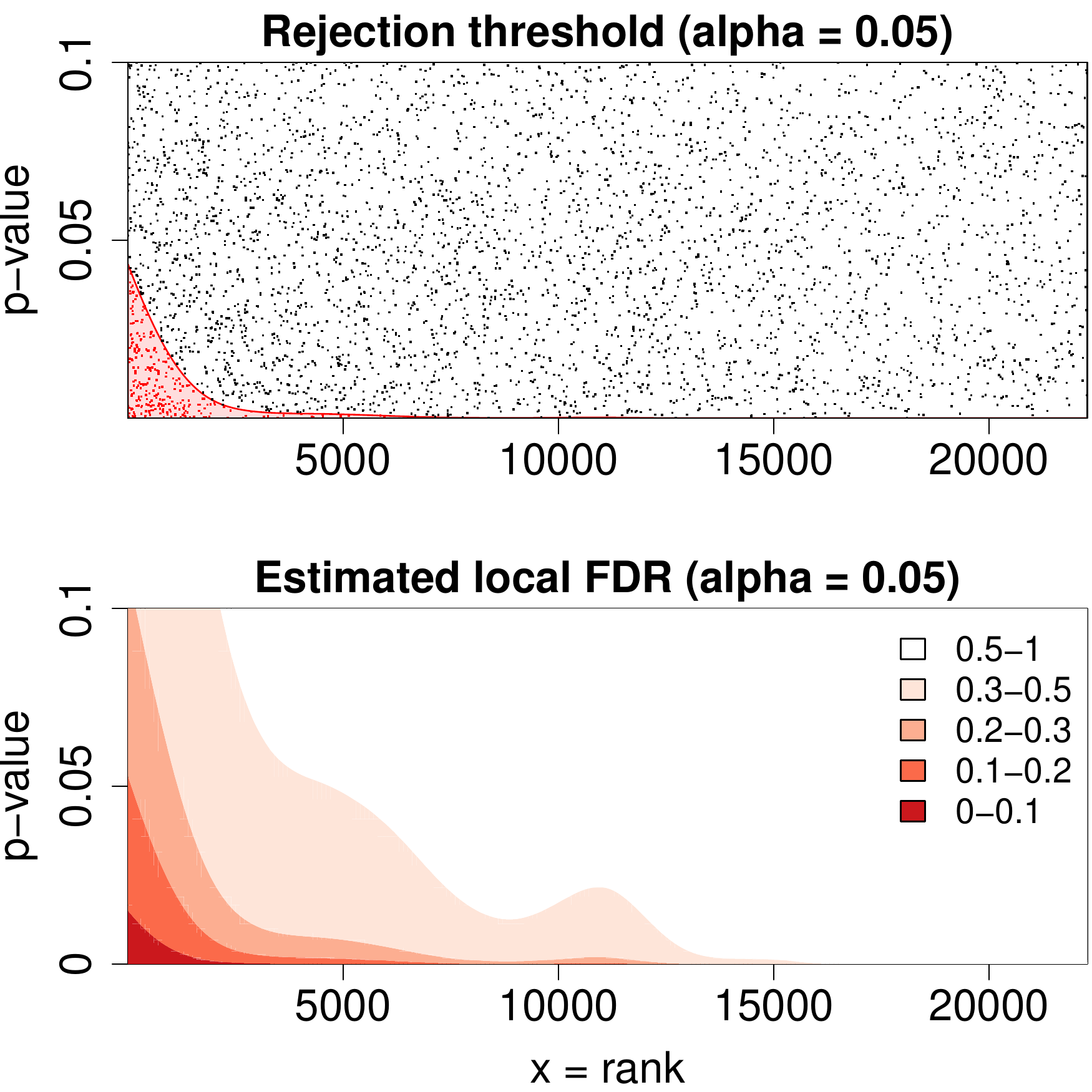}
  \includegraphics[width = 0.49\textwidth]{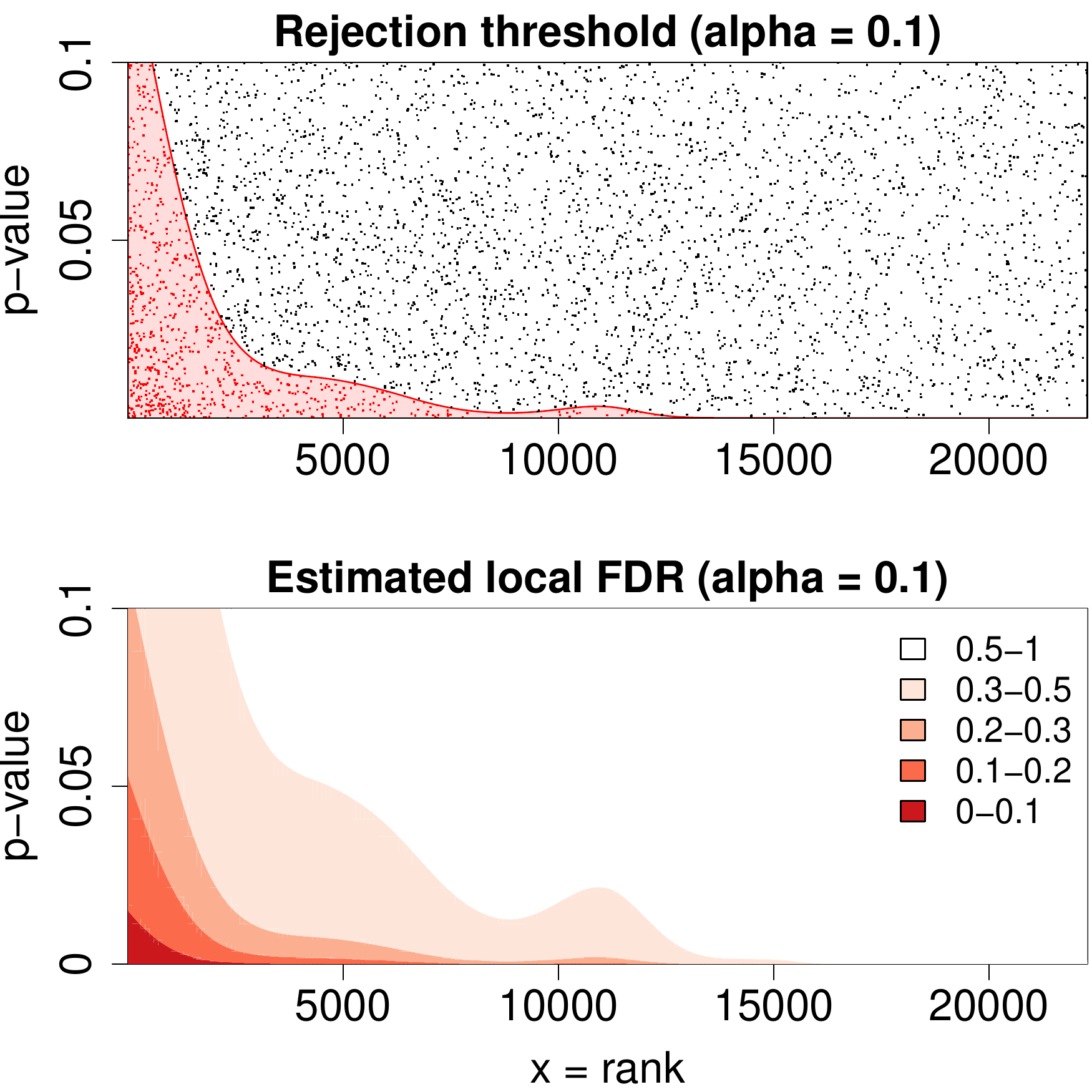}
  \caption{Results for gene/drug response data with moderately informative ordering of p-values, i.e. $\{p_{i}^{W}\}$, with $\alpha = 0.05$ (left) and $\alpha = 0.1$ (right): (top) the dots represent the p-values and the red dots are rejected ones. The red curve is the thresholding rule $s(x)$; (bottom) the contour plots of estimated local FDR.}\label{fig:moderate}
\end{figure}

The right two panels of Figure \ref{fig:GEOquery_rejs} correspond to the weaker and the strong orderings, and show that AdaPT significantly outperforms all other methods for all target FDR levels. One might doubt whether the power gain is driven by overfitting. To check this, we also apply AdaPT, as well as all other methods, on the same set of p-values with a random ordering. We repeat it using 100 random seeds and report the average number of rejections in the left panel of Figure \ref{fig:GEOquery_rejs}. In this case, the number of rejections drop dramatically and the power is almost the same as Barber-Cand{\`e}s method, the non-adaptive version of AdaPT. This provides strong evidence against overfitting.



To illustrate how \adapt exploits the covariate to improve the power, we plot the thresholding rules and estimated signal strength for p-values with moderately informative ordering and p-values with highly informative ordering, respectively in Figure \ref{fig:moderate} and Figure \ref{fig:high}. It can be seen from the bottom panels that the evidence to be non-null has an obvious decreasing trend when the ordering is used. Moreover, the highly informative ordering indeed sorts the p-values better than the moderately informative ordering. For the former, the thresholding rule is fairly monotone while it has a small bump at $i\approx 5000$ for the latter. In both cases, most discoveries are from the first 5000 genes in the list.

\begin{figure}[H]
  \centering
  \includegraphics[width = 0.49\textwidth]{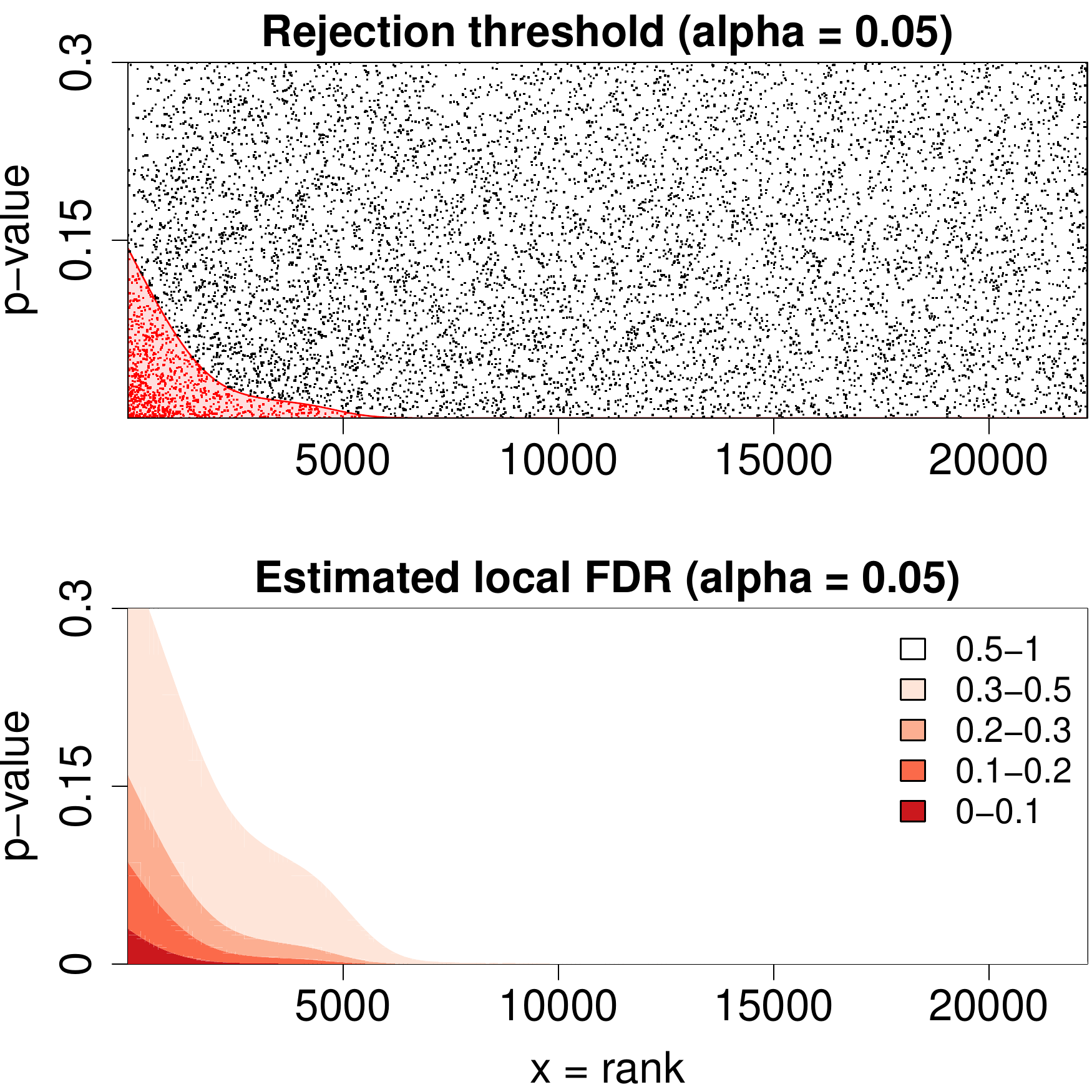}
  \includegraphics[width = 0.49\textwidth]{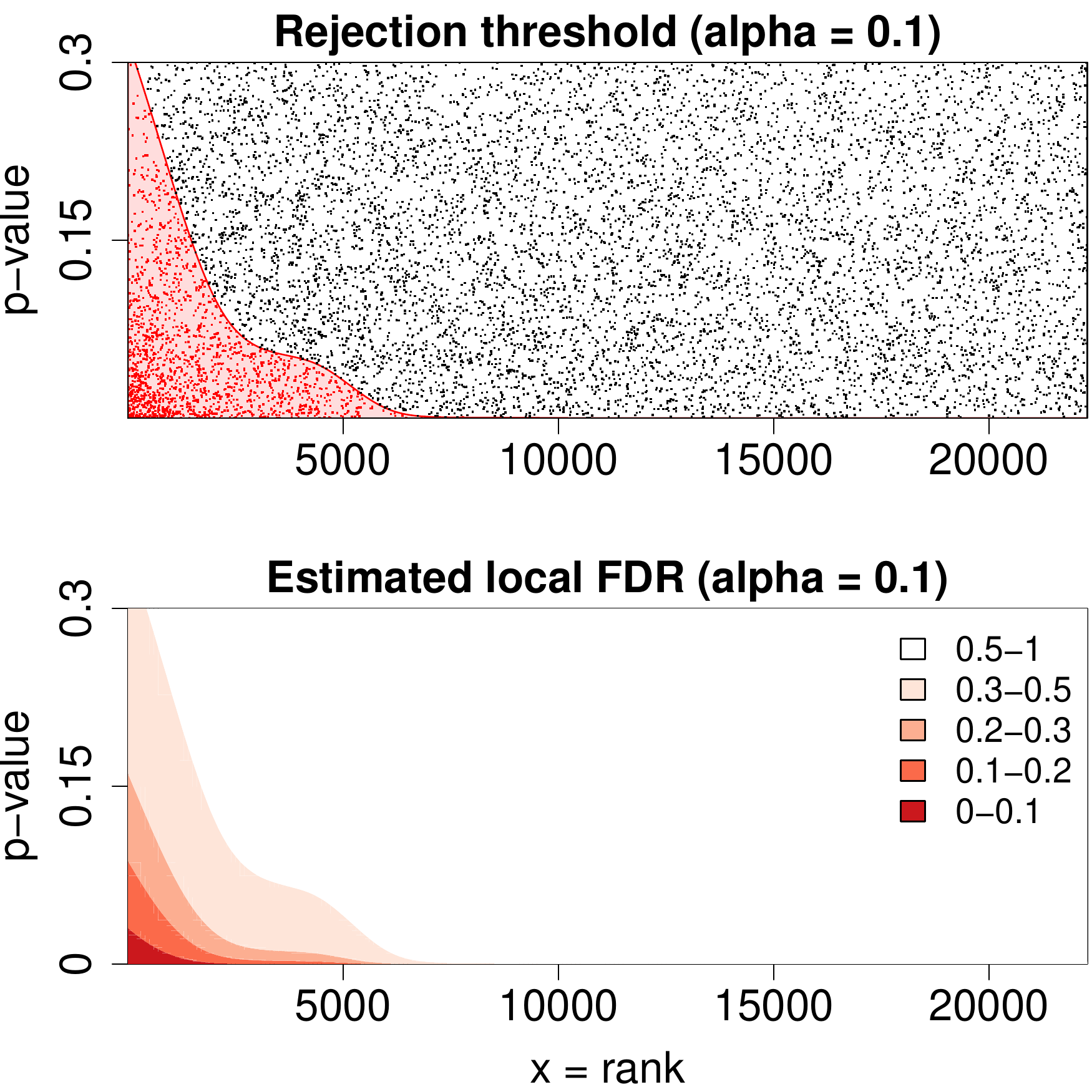}
  \caption{Results for gene/drug response data with highly informative ordering of p-values, i.e. $\{p_{i}^{S}\}$, with $\alpha = 0.05$ (left) and $\alpha = 0.1$ (right): (top) the dots represent the p-values and the red dots are rejected ones. The red curve is the thresholding rule $s(x)$; (bottom) the contour plots of estimated local FDR.}\label{fig:high}
\end{figure}

Finally, we measure the information loss caused by partial masking: We first estimate local FDR using the set of (unmasked) p-values and the covariates, denoted by $\mathrm{lfdr}^{*}(x)$. It can be regarded as the best possible estimate given the algorithm. Let $\mathrm{lfdr}_t(x)$ denote the estimate of local FDR at step $t$ (based on partially masked p-values). Then we measure the information loss by the correlation of $\{\mathrm{lfdr}^{*}(x_i)\}_{i=1}^{n}$ and $\{\mathrm{lfdr}_{t}(x_i)\}_{i=1}^{n}$. The results are shown in Figure \ref{fig:GEOquery_corr} where the x-axis corresponds to the target FDR, in a reverse order ranging from $0.5$ to $0.01$, and y-axis corresponds to the correlation at the step where $\hFDP$ first drops below the target FDR. As expected from the discussion in Subsection \ref{subsec:framework}, the information loss is quite small and even negligible after the target $\FDR$ drops to the ``practical'' regime (e.g. below $0.2$), where the correlation between $\{\mathrm{lfdr}^{*}(x_i)\}_{i=1}^{n}$ and $\{\mathrm{lfdr}_{t}(x_i)\}_{i=1}^{n}$ is almost $1$. The pattern is even more significant in other data examples in the next Subsection. This provides a strong evidence that \adapt allows efficient data exploration under comparatively limited information loss.

In summary, these plots show a strong data adaptivity of \adapt, which can also learn the local structure of data while controlling FDR. Moreover, it provides a quantitative way, by estimated signal strength, to evaluate the quality of ordering, which is the major concern in ordered testing problems \citep{li2016accumulation, lei2016power, li2016multiple}.

\begin{figure}[H]
  \centering
  \includegraphics[width = 0.55\textwidth]{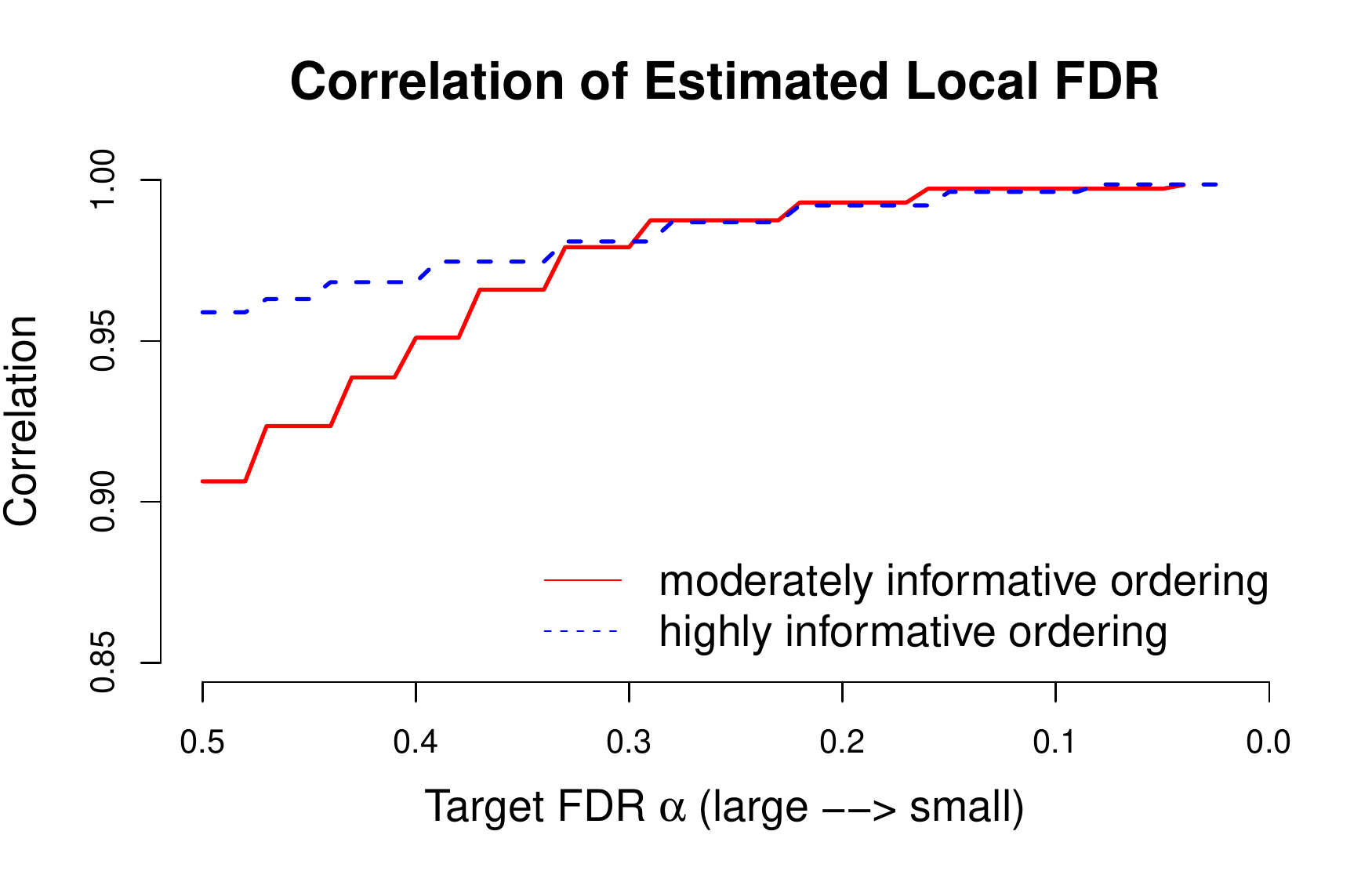}
  \caption{Correlation of $\{\mathrm{lfdr}^{*}(x_i)\}_{i=1}^{n}$ and $\{\mathrm{lfdr}_{t}(x_i)\}_{i=1}^{n}$ for gene/drug-response dosage dataset under original, moderately informative and highly informative orderings. The x-axis corresponds to the target FDR, in a reverse order ranging from $0.5$ to $0.01$, and y-axis corresponds to the correlation at the step where $\hFDP$ first drops below the target FDR. }\label{fig:GEOquery_corr}
\end{figure}

\subsection{Simulation studies}

$\bullet$ \textbf{Example 1: a two-dimensional case}

~\\
\noindent We generate the covariates $x_i$'s from an equi-spaced $50\times 50$ grid in the area $[-100, 100]\times [-100, 100]$. We generate $p$-values i.i.d. from a one-sided normal test, i.e. 
\begin{equation}\label{eq:simul_pvals}
p_{i} = 1 - \Phi(z_{i}), \quad \mbox{and} \quad z_{i}\sim N(\mu, 1),
\end{equation}
where $\Phi$ is the cdf of $N(0, 1)$. For $i\in \cH_{0}$ we set $\mu = 0$ and for $i\not\in \cH_{0}$ we set $\mu = 2$. Figure~\ref{fig:simul1_truth} below shows three types of $\cH_{0}$ that we conduct tests on.

\begin{figure}[H]
  \centering
  \includegraphics[width = 0.75\textwidth]{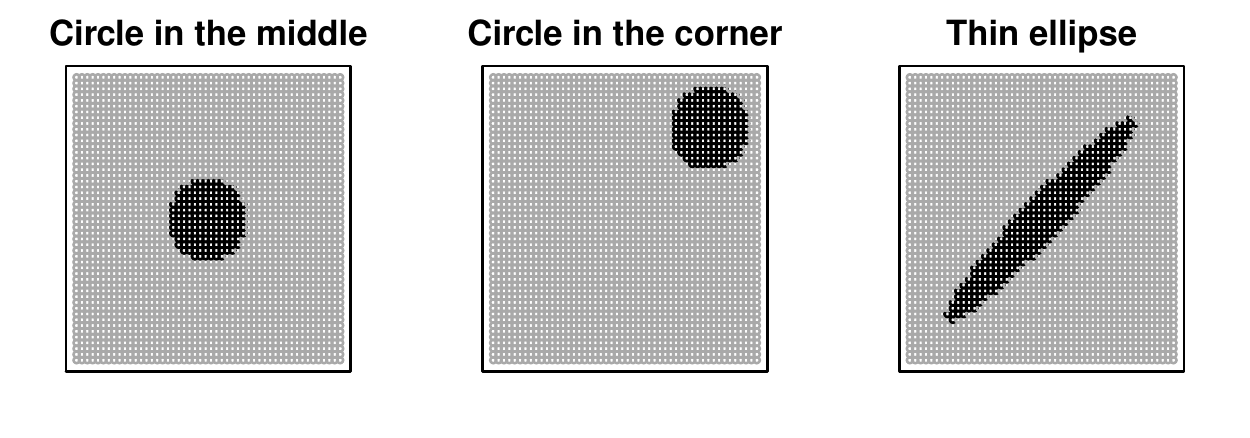}
  \caption{The above panels display the underlying ground truth for three cases in Example 1. Each point represents a hypothesis (2500 in total) with gray ones being nulls and black ones being non-nulls.}\label{fig:simul1_truth}
\end{figure}

In this case, it is not clear how to apply non-adaptive ordered testing procedures or Independent Filter. Thus we compare \adapt only with Storey's BH method, Barber-Cand\'{e}s method, IHW using the default automatic parameter tuning procedure and SABHA using 2-dim low total variation weights (see Section 4.3 of \cite{li2016multiple}). For \adapt, we fit two-dimensional Generalized Additive Models in M-step, using \pkg{R} package \pkg{mgcv} with the knots selected automatically in every step by GCV criterion. For each procedure and a given level $\alpha$, let $\cR_{\alpha}$ be the set of rejected hypotheses with a target FDR level $\alpha$. Then we calculate the FDP and the power as
\begin{align}\label{eq:defs}
\mathrm{FDP}(\alpha) &= \frac{|\cR_{\alpha}\cap \cH_{0}|}{|\cR_{\alpha}|}, \quad \mathrm{power}(\alpha) ~= \frac{|\cR_{\alpha}\cap \cH_{0}^{c}|}{|\cH_{0}^{c}|}.
\end{align}
We repeat the above procedure for on $100$ fresh simulated datasets and calculate the average of $\mathrm{FDP}(\alpha)$ and $\mathrm{power}(\alpha)$ as the measure of FDR and power. The results are shown in Figure \ref{fig:simul1}. It is clearly seen that \adapt controls FDR as other methods while achieving a significantly higher power. 

To see why \adapt gains power, we plot the estimated local FDR in Figure \ref{fig:simul1_fit} for the first case, at the initial step, the step where $\hFDP$ is first below 0.3 and the step where $\hFDP$ is first below 0.1. As shown in the real examples, the fitted local FDR identifies the non-nulls quite accurately even at the initial step where most p-values are partially-masked. The estimates become very stable and informative after reaching the practical regime of $\alpha$'s.

\begin{figure}[H]
  \centering
  \includegraphics[width = 0.85\textwidth]{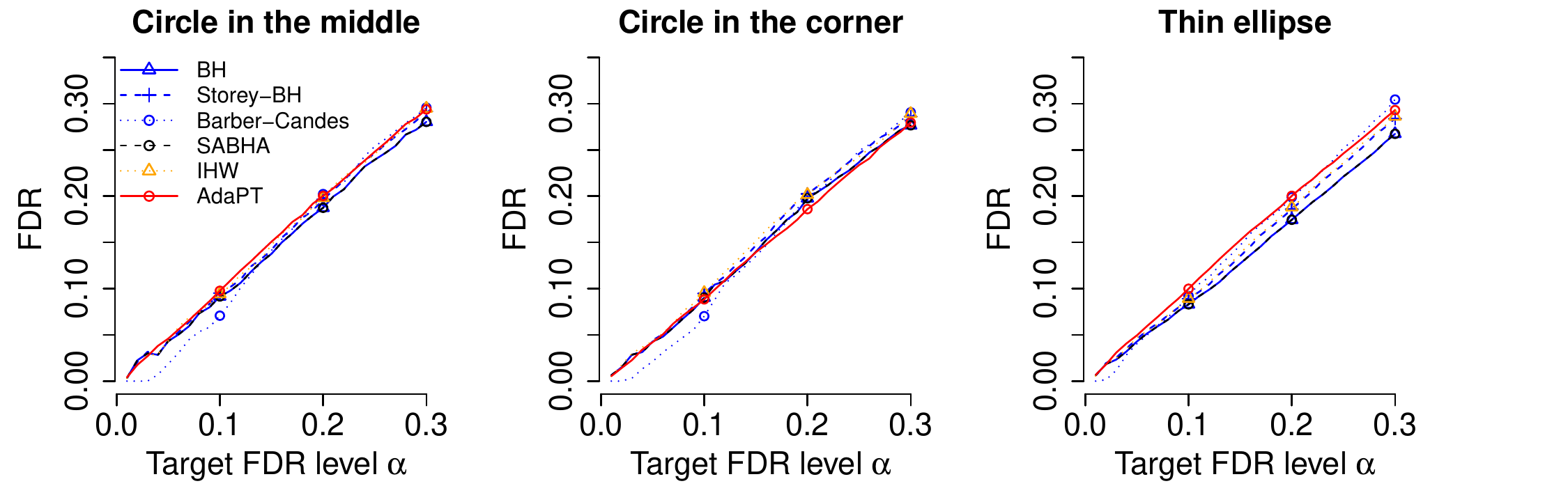}
  \includegraphics[width = 0.85\textwidth]{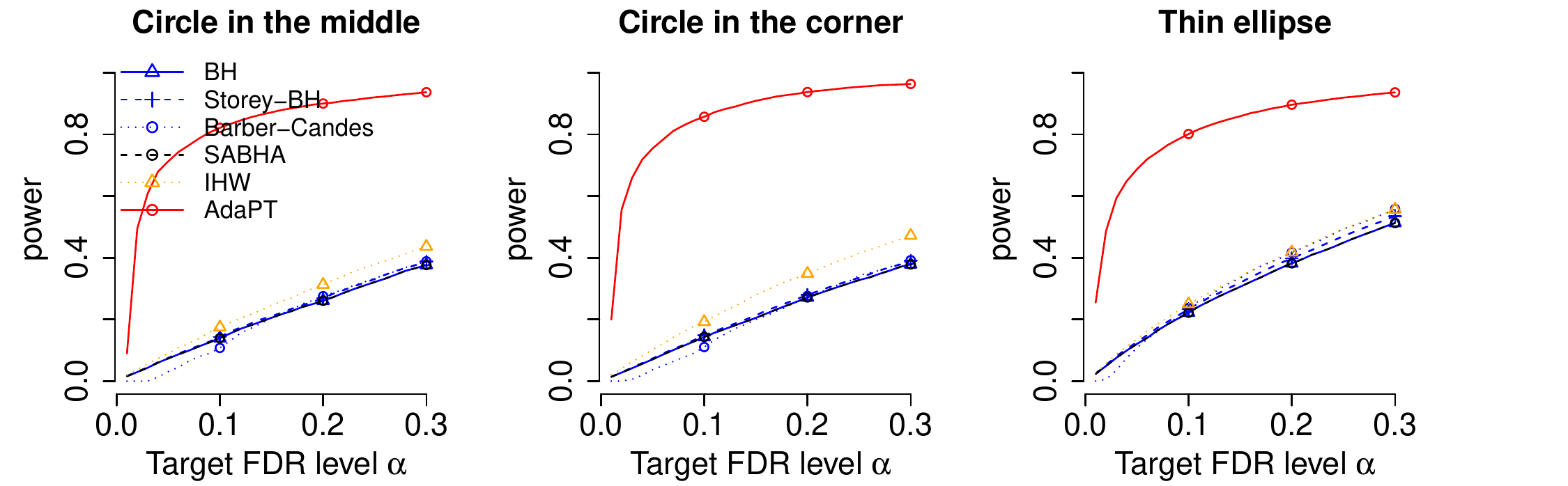}
  \caption{FDR and power with $\alpha \in \{0.01, 0.02, \ldots, 0.30\}$ in Example 1.}\label{fig:simul1}
\end{figure}

\begin{figure}[H]
  \centering
  \includegraphics[width = 0.85\textwidth]{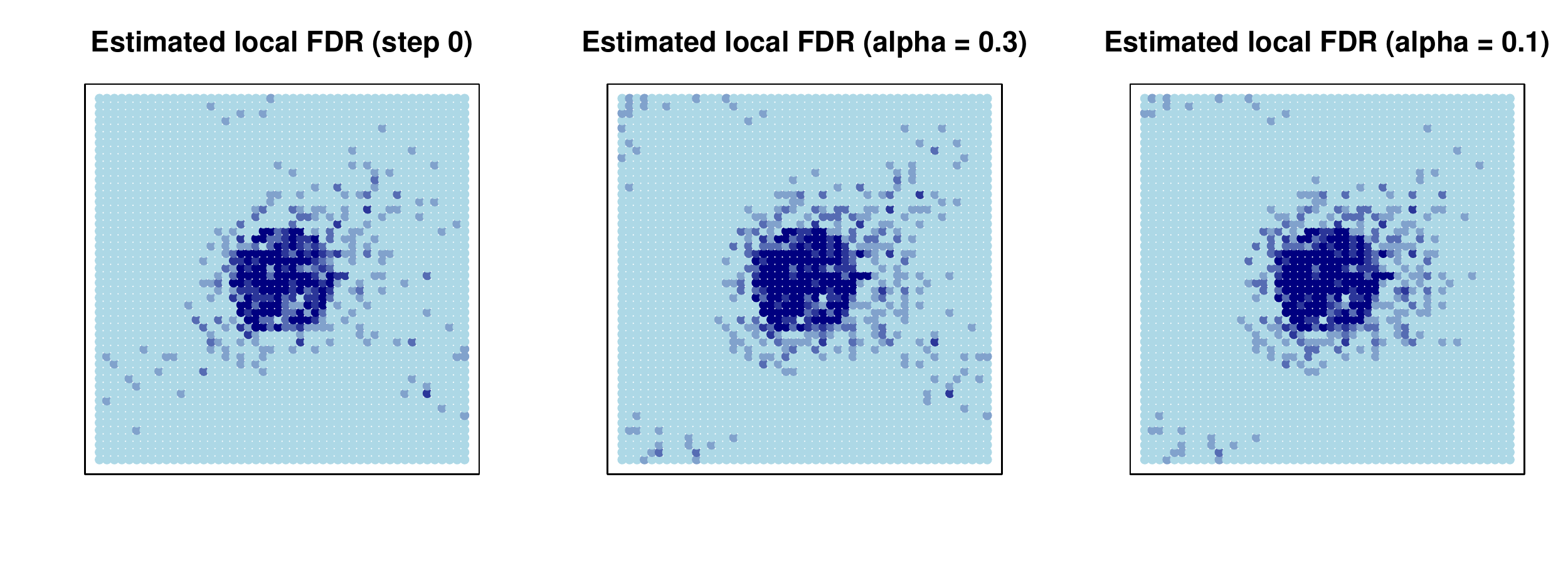}
  \caption{Estimated local FDR in the first case of Example 1 at the initial step (left), with the target FDR level $0.3$  (middle) and with the target FDR level $0.1$ (right). The dark color marks the hypotheses with low local FDR and vice versa.}\label{fig:simul1_fit}
\end{figure}

\noindent $\bullet$ \textbf{Example 2: a 100-dimensional case}

~\\
\noindent We generate $x_{i}\in \R^{d}$ with $d = 100$ and 
\[\{x_{ij}: i\in [n], j\in [d]\}\stackrel{i.i.d.}{\sim} U([0,1]).\]
Then we generate p-values from a varying-coefficient two group beta-mixture model \eqref{eq:beta_mixture} with $\pi_{1i}$ and $\mu_{i}$ are specified as a logistic model and a truncated linear model, respectively, i.e.,
\[\log \lb\frac{\pi_{1i}}{1 - \pi_{1i}}\rb = \theta_{0} + x_i^{T}\theta, \quad \mu_i = \max\{x_i^{T}\beta, 1\}, \quad \beta,\theta\in \R^{d}.\]
In this case, we choose $\theta$ and $\beta$ as highly sparse vectors with only two non-zero entries: 
\[\theta = (3, 3, 0, \ldots, 0)^{T}, \quad \beta = (2, 2, 0, \ldots, 0)^{T}\]
and $\theta_{0}$ is chosen so that $\frac{1}{n}\sum_{i=1}^{n}\pi_{1i} = 0.3$. In this case, $\E (-\log p_i) = \mu_i$ under the alternative. Figure \ref{fig:simul2_hist} shows the histograms of $\pi_{1i}$'s and $\mu_i$'s. 

\begin{figure}[H]
  \centering
  \includegraphics[width = 0.85\textwidth]{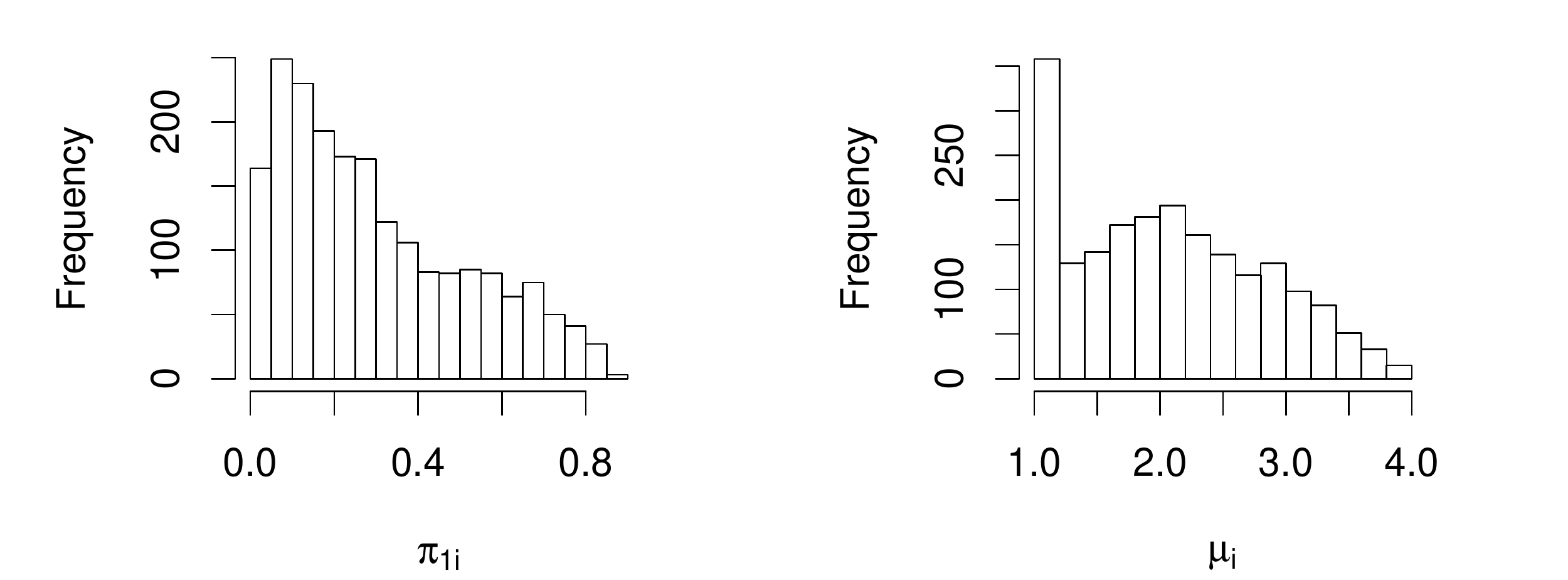}
  \caption{Distributions of $\pi_{1i}$'s and $\mu_i$'s in Example 2.}\label{fig:simul2_hist}
\end{figure}

In this case, it is not clear how to apply non-adaptive ordered testing procedures or Independent Filter or adaptive procedures like IHW and SABHA. Thus we compare \adapt only with BH method, Storey's BH method and Barber-Cand\'{e}s method. For \adapt, we fit $L_1$-regularized GLMs in M-step (See Appendix \ref{app:EM} for details), using \pkg{R} package \pkg{glmnet} with the penalty level selected automatically in every step by cross validation. Further we run \adapt by fitting an "oracle" GLM in M-steps where only the first two covariates are involved. 

As in Example 1, we estimate the FDR and the power using 100 replications. The results are plotted in Figure \ref{fig:simul2}. It is clearly seen that both \adapt 's control FDR as other methods while achieving a higher power. Not surprisingly, compare to \adapt with $L_1$-regularized GLMs, \adapt with ``oracle'' GLMs has a higher power. Nevertheless, this example shows the unprecedented ability of \adapt to improve power by squeezing information from a large set of noisy features.

\begin{figure}[H]
  \centering
  \includegraphics[width = 0.85\textwidth]{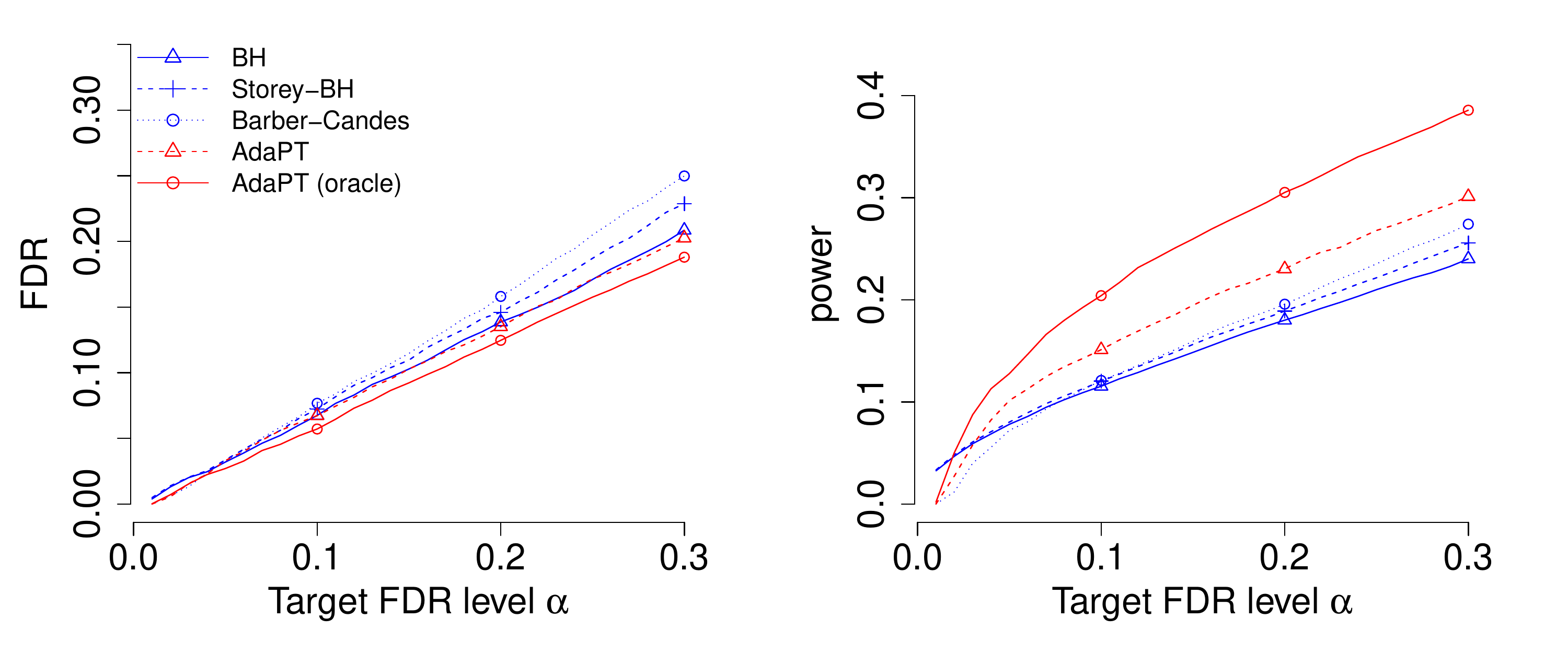}
  \caption{FDR and power with $\alpha \in \{0.01, 0.02, \ldots, 0.30\}$ in Example 2.}\label{fig:simul2}
\end{figure}

\subsection{Other exemplary applications}
In this Subsection, we examine the performance of \adapt on four more real datasets, which are analyzed in other papers exploiting adaptive FDR control methods, e.g. \cite{bourgon2010independent, ignatiadis2016data}. In all cases, we start with a brief introduction of the dataset and show the plots on the number of rejections, as Figure \ref{fig:GEOquery_rejs}, path of information loss, as Figure \ref{fig:GEOquery_corr}, and threshold curve and level curves of estimated local FDR with target FDR $0.1$, as Figure \ref{fig:moderate} and Figure \ref{fig:high}. We use the same settings for \adapt as in the gene dosage dataset: performing model selection at the initial step with candidate featurization being all combinations of spline basis with $6\sim 15$ equi-quantile knots on $\pi(x)$ and $\mu(x)$; and fixing the selected model in subsequent updates.

~\\
$\bullet$ \textbf{Bottomly data}

~\\
This dataset is an RNA-Seq dataset targeting on detecting the differential expression on two mouse strains, C57BL/6J (B6) and DBA/2J (D2), collected by \cite{bottomly}, available on ReCount repository \citep{recount}, and analyzed by \cite{ignatiadis2016data} using IHW. It consists of gene expression measurements for $n = 13932$ genes. Following \cite{ignatiadis2016data}, we analyze the data using \pkg{DEseq2} package \citep{DEseq2} and use the logarithm of normalized count (averaged across all samples plus 1) as the univariate covariate for each gene. The results are plotted in Figure \ref{fig:bottomly}. It is clearly seen that \adapt produces significantly more discoveries than all other methods and the information loss is almost negligible (with correlation consistently above 0.985). Furthermore, we observe the same pattern that \adapt prioritizes the genes with higher mean normalized means.

\begin{figure}[H]
  \centering
  \includegraphics[width = 0.3\textwidth]{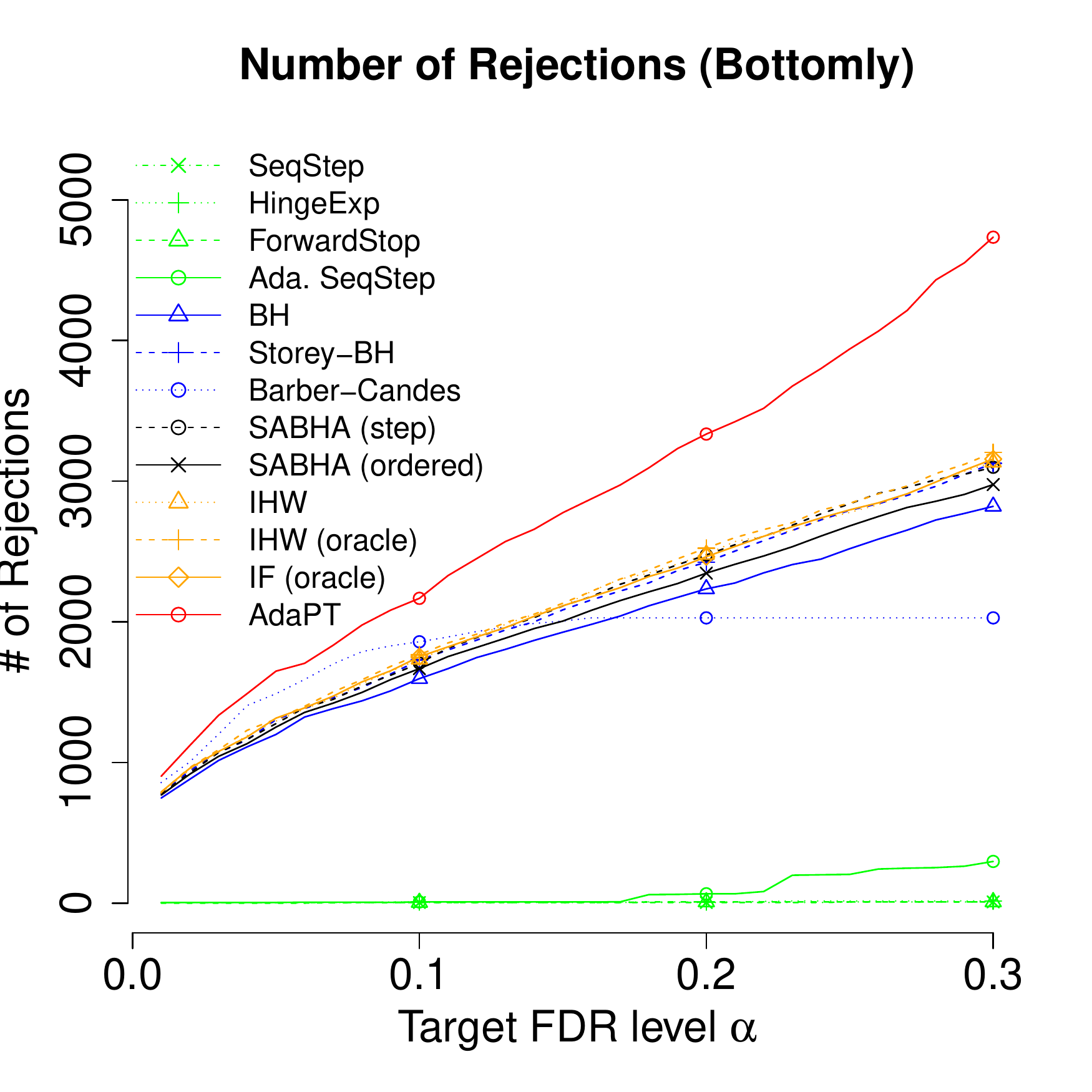}
  \includegraphics[width = 0.3\textwidth]{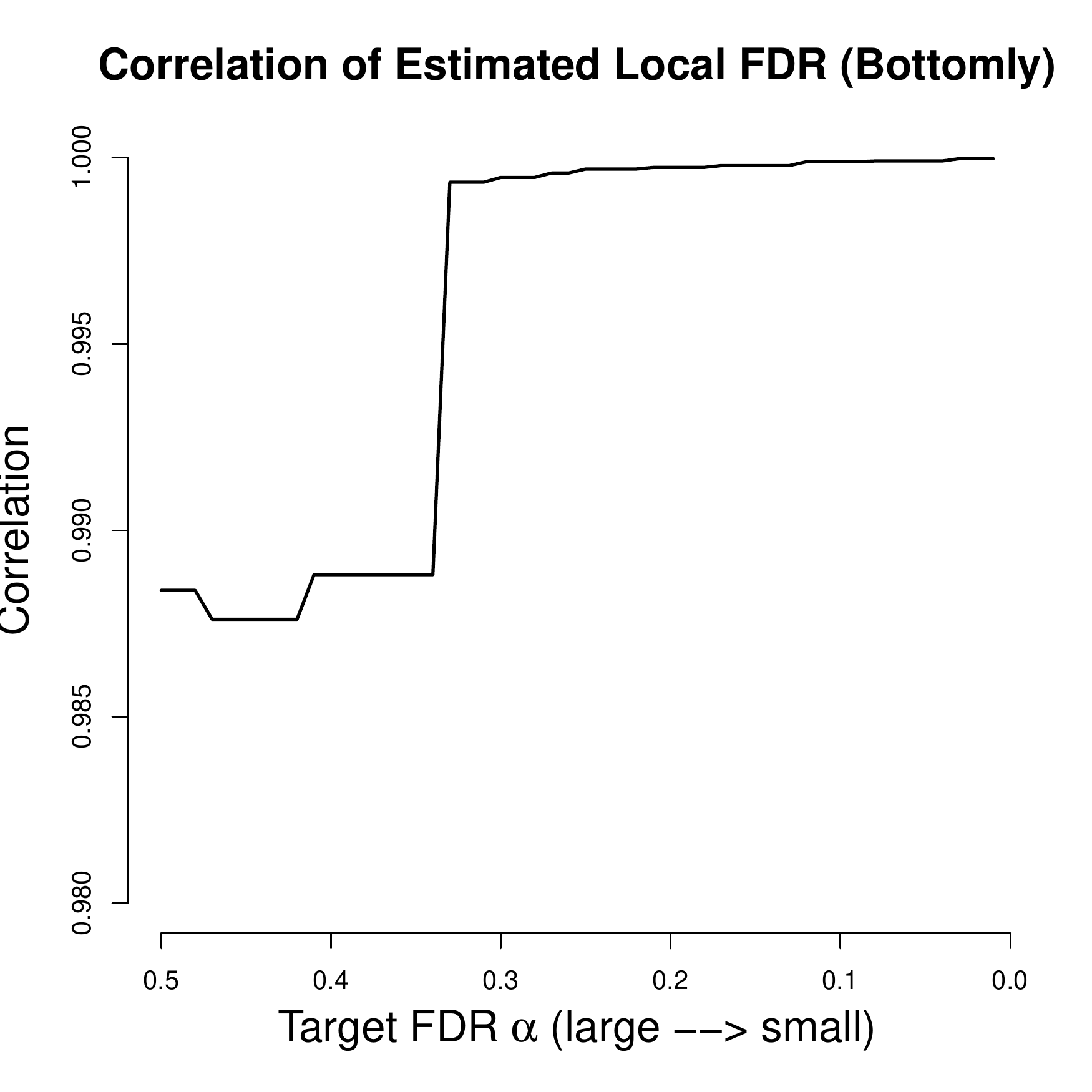}
  \includegraphics[width = 0.3\textwidth]{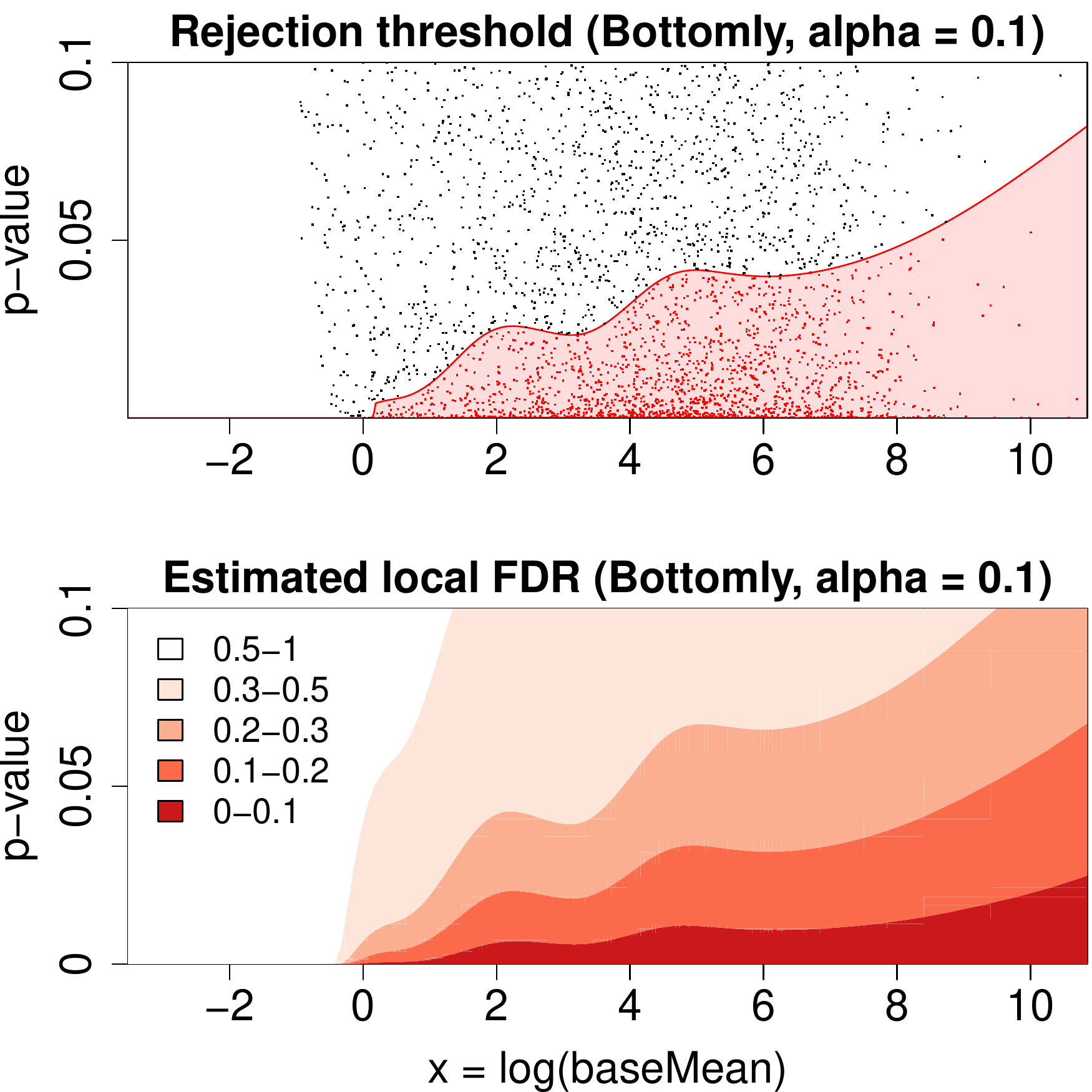}
  \caption{Results for Bottomly dataset: (left) number of rejections; (middle) path of information loss; (right) threshold curve and level curves of estimated local FDR when $\alpha = 0.1$.}\label{fig:bottomly}
\end{figure}

\noindent $\bullet$ \textbf{Airway data}

~\\
\noindent This dataset is an RNA-Seq dataset targeting on identifying the differentially expressed genes in airway smooth muscle cell lines in response to dexamethasone, collected by \cite{airway} and available in \pkg{R} package \pkg{airway}. It is analyzed in the vignette of \pkg{IHW} package using IHW method \cite{ignatiadis2016data}. As in the vignette and the previous example, we analyze the data using \pkg{DEseq2} package \citep{DEseq2} and use the logarithm of normalized count as the univariate covariate for each gene. The results are plotted in Figure \ref{fig:airway}. Again, \adapt produces significantly more discoveries than all other methods. 

\begin{figure}[H]
  \centering
  \includegraphics[width = 0.3\textwidth]{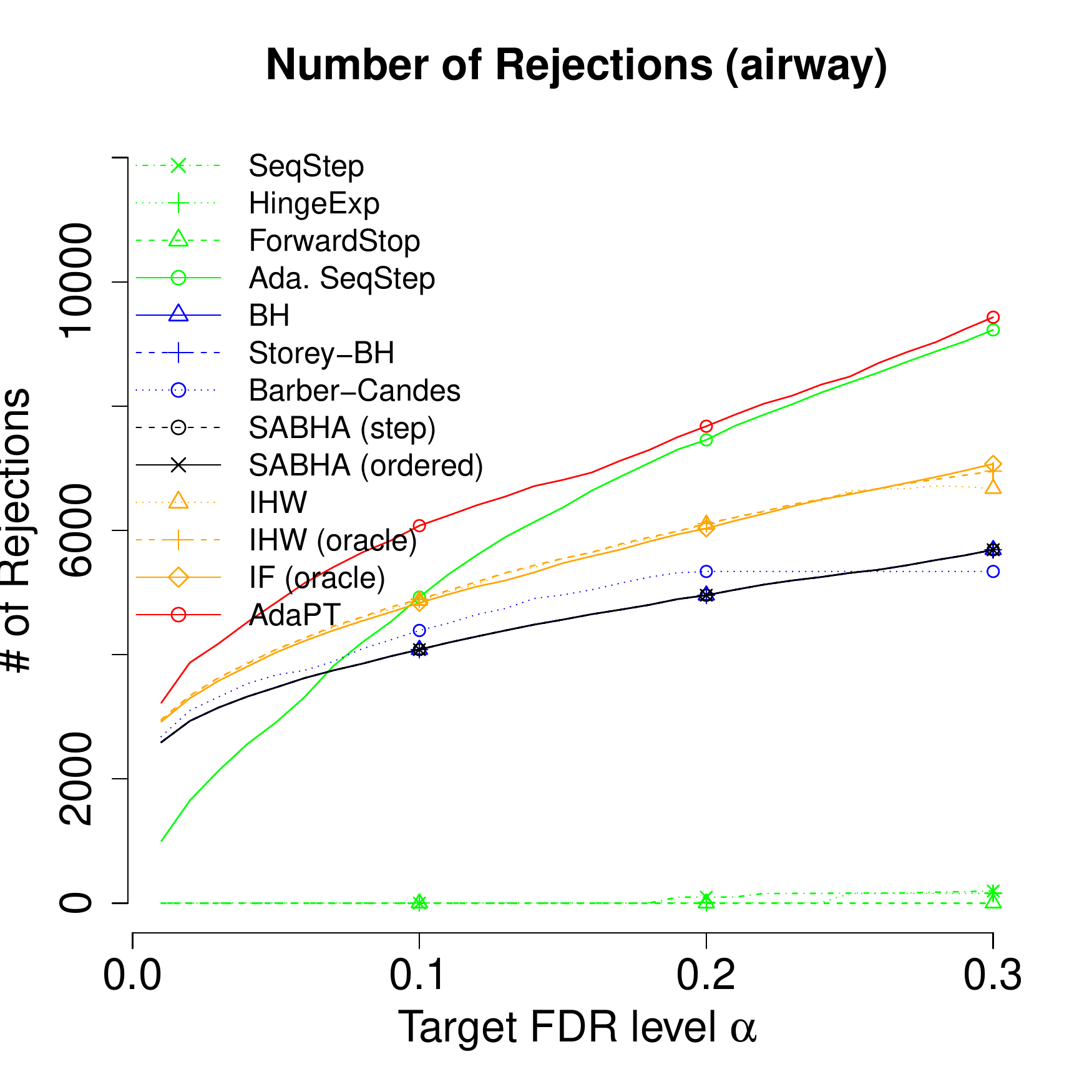}
  \includegraphics[width = 0.3\textwidth]{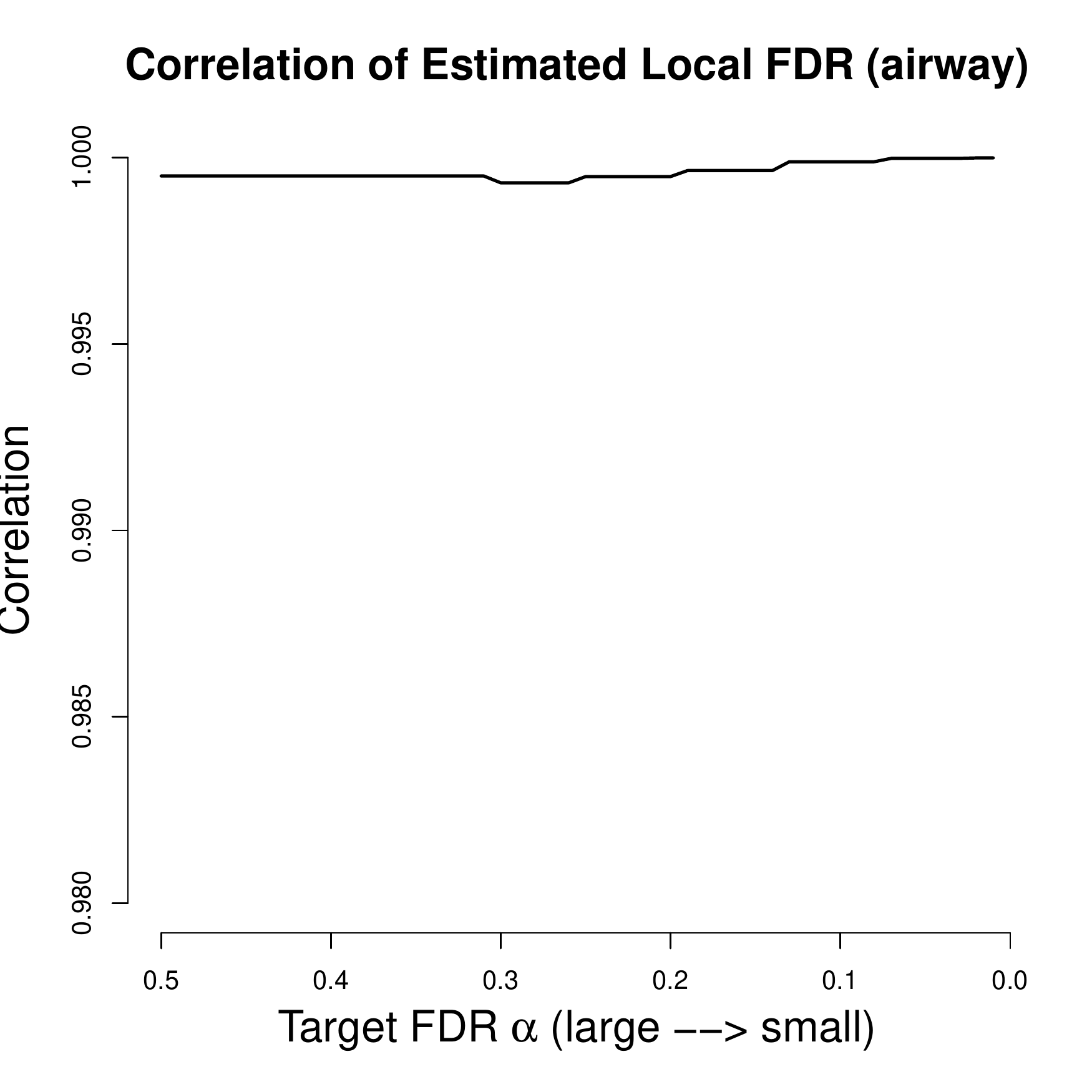}
  \includegraphics[width = 0.3\textwidth]{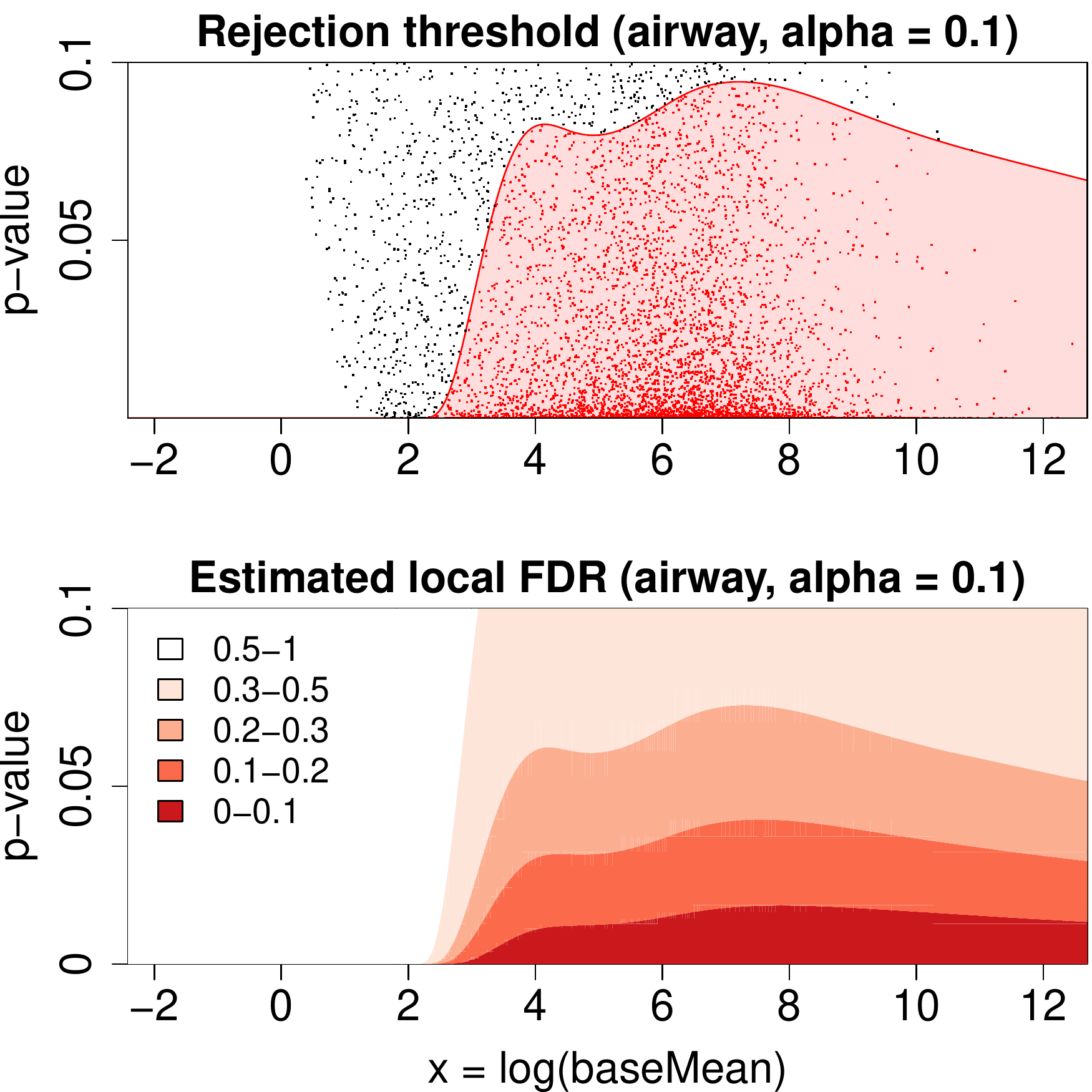}
  \caption{Results for Airway dataset: (left) number of rejections; (middle) path of information loss; (right) threshold curve and level curves of estimated local FDR when $\alpha = 0.1$.}\label{fig:airway}

\end{figure}

\noindent $\bullet$ \textbf{Pasilla data}

~\\
This dataset is also an RNA-Seq dataset targeting on detecting genes that are differentially expressed between the normal and Pasilla-knockdown conditions, collected by \cite{pasilla} and available in \pkg{R} package \pkg{pasilla} \citep{pasillapackage}. It is analyzed in the vignette of \pkg{genefilter} package \citep{genefilter} using independent filtering  method \cite{bourgon2010independent}. As in the vignette, we analyze the data using \pkg{DEseq} package \citep{DEseq} and use the logarithm of normalized count as the univariate covariate for each gene. The results are plotted in Figure \ref{fig:pasilla}. It is clear that we arrive at the same conclusion that  \adapt is more powerful than all other methods.

\begin{figure}[H]
  \centering
  \includegraphics[width = 0.3\textwidth]{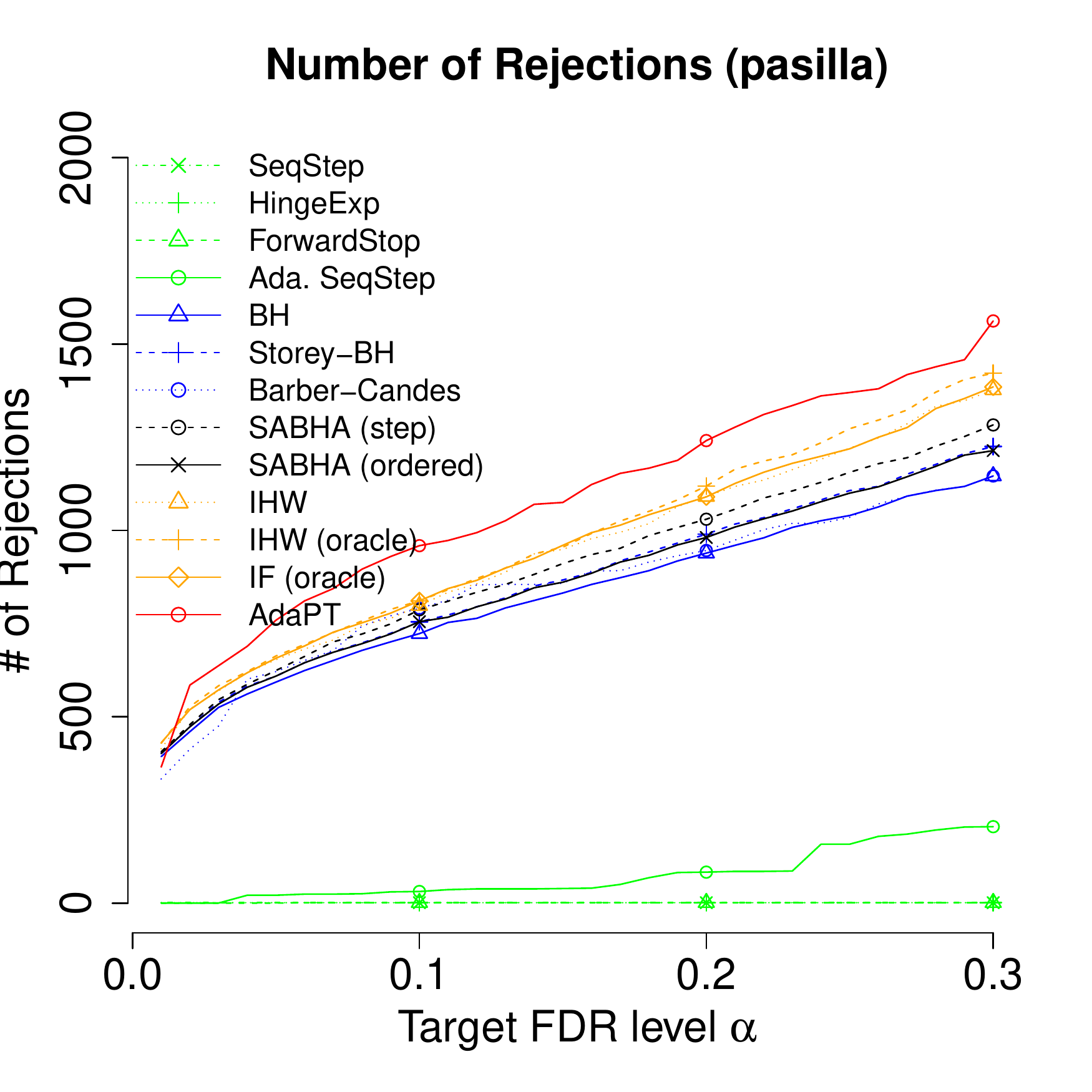}
  \includegraphics[width = 0.3\textwidth]{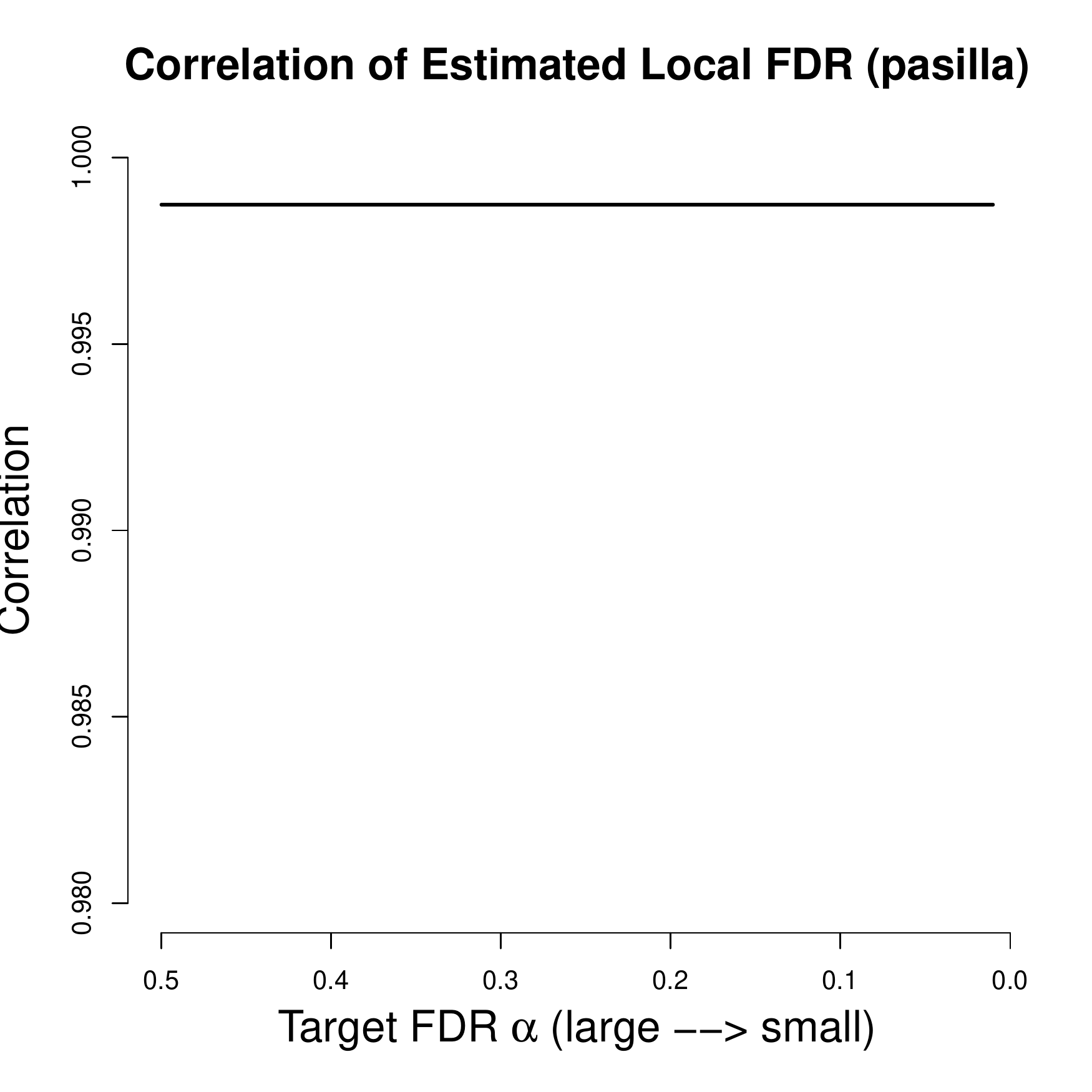}
  \includegraphics[width = 0.3\textwidth]{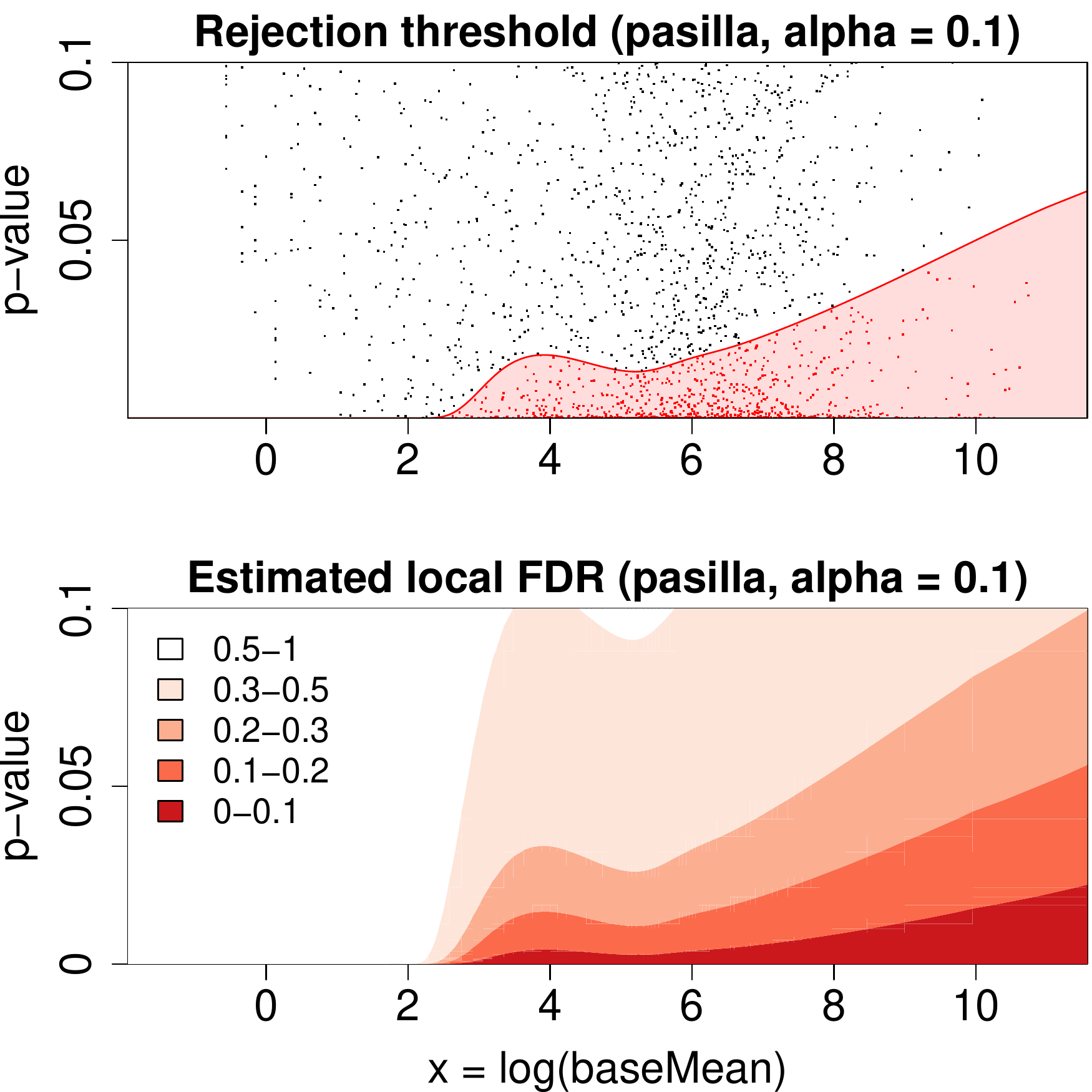}  
  \caption{Results for Pasilla dataset: (left) number of rejections; (middle) path of information loss; (right) threshold curve and level curves of estimated local FDR when $\alpha = 0.1$.}\label{fig:pasilla}

\end{figure}

\noindent $\bullet$ \textbf{Yeast proteins data}

~\\
This dataset is a proteomics dataset, collected by \cite{SILAC} and available in \pkg{R} package \pkg{IHWpaper}, that provides temporal abundance profiles for 2666 yeast proteins from a quantitative mass-spectrometry (SILAC) experiment. The goal is to identify the differential protein abundance in yeast cells treated with rapamycin and DMSO. It is analyzed in \cite{ignatiadis2016data} using IHW method. As in \cite{SILAC} and \cite{ignatiadis2016data}, we calculate the p-values using Welch's t-test and use as the univariate covariate the logarithm of total number of peptides that were quantified across all samples for each gene. The results are plotted in Figure \ref{fig:proteomics}. In this case, \adapt has a similar performance to Barber-Cand\'{e}s method and Storey's BH method. However, it still outperforms all other methods. Furthermore, \adapt learns the monotone pattern of the local FDR, which coincides with the heuristic. 

\begin{figure}[H]
  \centering
  \includegraphics[width = 0.3\textwidth]{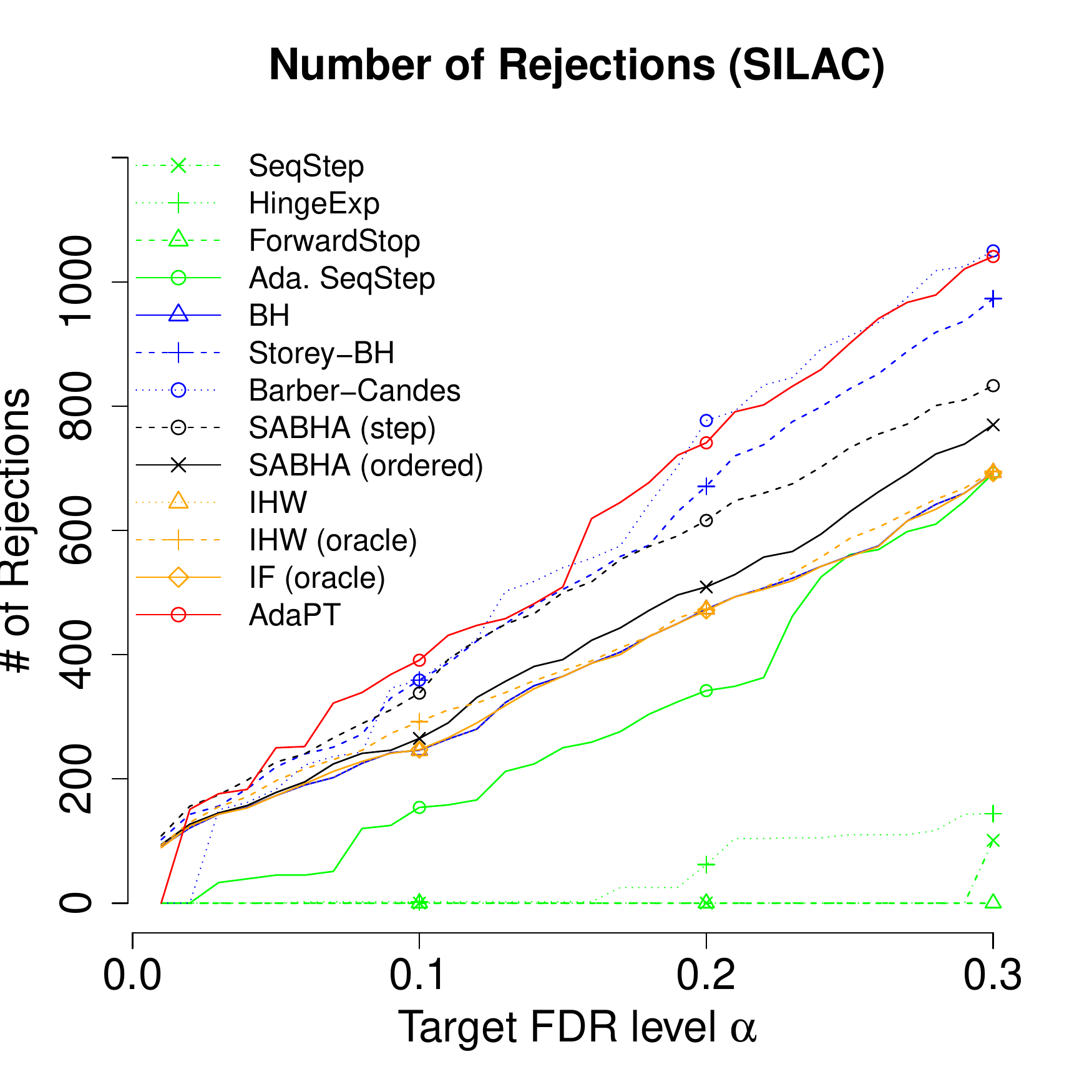}
  \includegraphics[width = 0.3\textwidth]{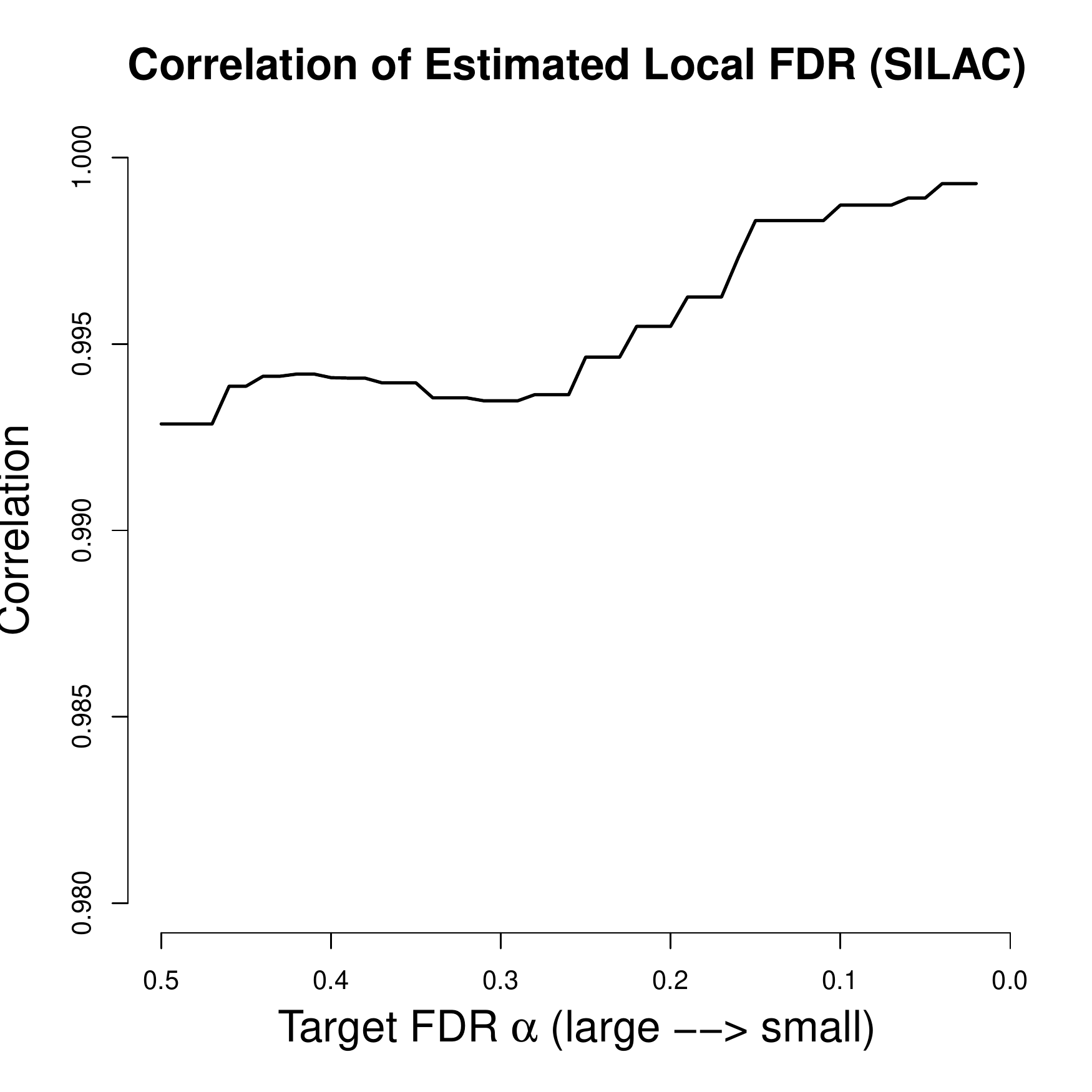}
  \includegraphics[width = 0.3\textwidth]{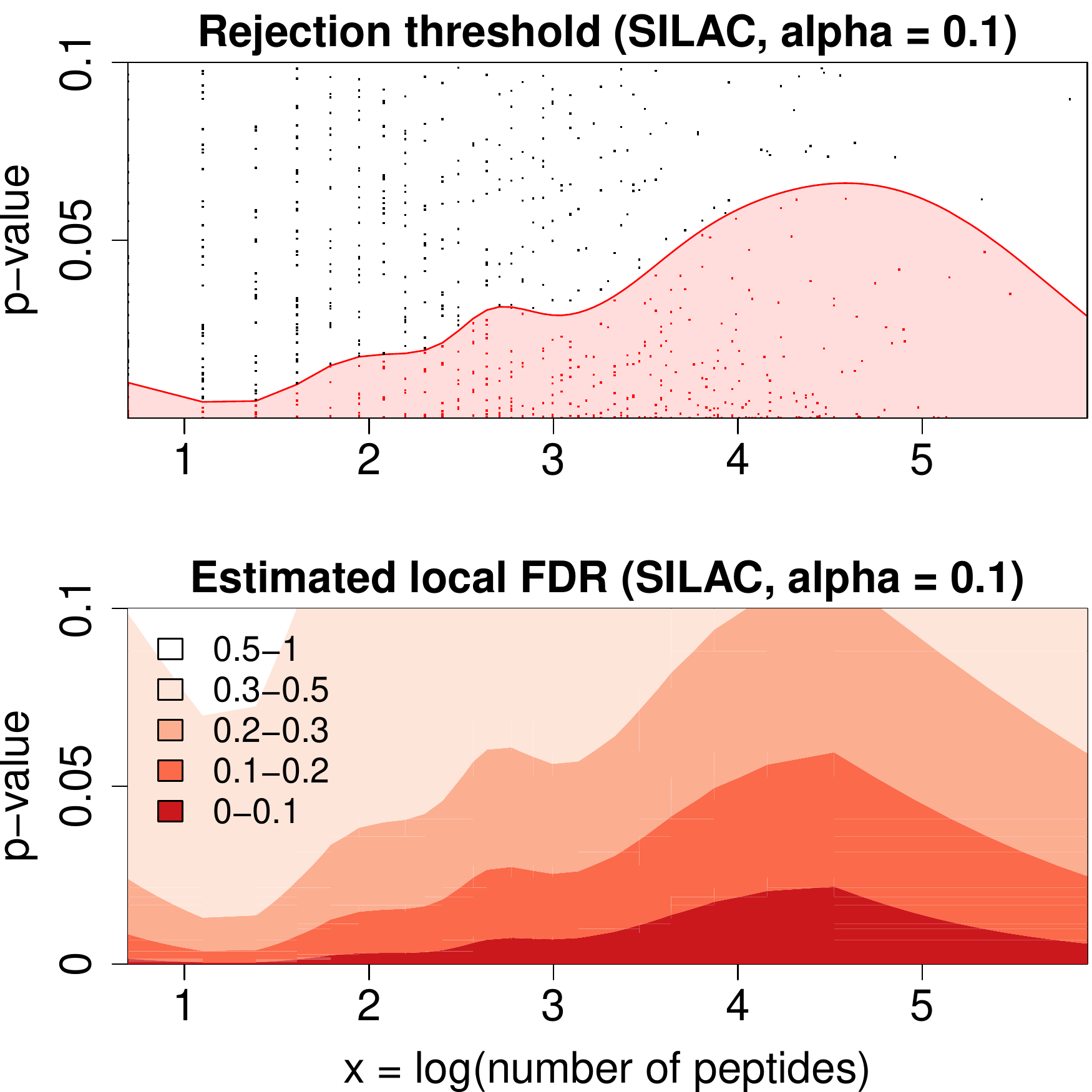}
  \caption{Results for yeast proteins dataset: (left) number of rejections; (middle) path of information loss; (right) threshold curve and level curves of estimated local FDR when $\alpha = 0.1$.}\label{fig:proteomics}
\end{figure}

\section{Discussion}\label{sec:discussion}

We have proposed the AdaPT procedure, a general iterative framework for multiple testing with side information. Using partially masked $p$-values, we estimate a family of optimal and increasingly stringent rejection thresholds, which are level surfaces of the local FDR. We then monitor an estimator of FDP to decide which threshold to use, updating our estimates as we unmask more $p$-values and gain more information. 

Our method is interactive in that it allows the analyst to use an arbitrary method for estimating the local FDR, and to consult her intuition to change models at any iteration, even after observing most of the data. No matter what the analyst does or how badly she overfits the data, FDR is still controlled at the advertised level (though power could be adversely affected by overfitting).We show using various experiments that \adapt can give consistently significant power improvements over current state-of-the-art methods.

\subsection{\adapt without thresholds}\label{subsec:adapt_no_thresh}
Although we state \adapt as a procedure that interactively updates a covariate-variant threshold curve, the thresholds are not essential. In fact, Algorithm \ref{algo:adapt} can be modified as follows in the absence of $s(x)$.

\begin{algorithm}
  \caption{\adapt without thresholds}\label{algo:adapt_no_thresh}
  \textbf{Input: } predictors and $p$-values $\left(x_{i}, p_{i}\right)_{i \in [n]}$, target FDR level $\alpha$. 
  
  \textbf{Procedure: }
  \begin{algorithmic}[1]
    \State Initialize $\cR_{0} = [n]$;
    \For{$t = 0, 1, \ldots$}
    \State $R_t\gets \#\left\{i\in \cR_t: p_i\le \frac{1}{2}\right\}; A_t\gets \#\left\{i\in \cR_t: p_i > \frac{1}{2}\right\}$;
    \State $\widehat{\FDP}_{t}\gets \frac{1 + A_{t}}{R_{t}\vee 1}$;
    \If{$\widehat{\FDP}_{t}\le \alpha$}
    \State Reject $\cR_t$;
    \EndIf
    \State $\cR_{t+1}\gets 
    \update(\left(x_{i}, \tp_{t,i}\right)_{i \in [n]}, \cR_t)$;
    \EndFor
  \end{algorithmic}
\end{algorithm}

Rephrasing Algorithm \ref{algo:adapt_no_thresh}: we start from partially masking all p-values, yielding a ``candidate rejection set'' $\cR_{0} = [n]$, then apply arbitrary method to update $\cR_{t}$ directly. The FDP estimator (line 4) is defined in an essentially identical way as Algorithm \ref{algo:adapt}. It is easy to see that Algorithm \ref{algo:adapt} is a special case of Algorithm \ref{algo:adapt_no_thresh}. Perhaps strikingly, the proof of FDR control carries through to this general case without any modification. 

It is not hard to see that our implementation in Section \ref{sec:implementation} can be reformulated in a more simple and straightforward way: in each step we estimate local FDR for each partially-masked p-values and peel off $\delta$-proportion of them with highest estimated local FDR.

In principle, we can define any ``score'' that measures how "promising" each hypothesis is or how ``likely'' each hypothesis is non-null. A simple workflow based on Algorithm \ref{algo:adapt_no_thresh} is to peel off the hypotheses with least favorable ``scores'' and proceed with refitted ``scores'' by exploiting the revealed p-values. Heuristically, the most statistical meaningful ``score'' is local FDR, which is directly associated with our purpose. However, it arguably allows the framework of \adapt to be more general and flexible. For instance, we recently exploited this idea and develop a general framework for controlling FDR under structural constraints. We refer the readers to \cite{lei2017star} for more thoughts in this vein.

\subsection{Extension to dependent data using knockoffs}\label{sec:hybrid}

\newcommand{\tv}{\tilde{v}}

It would also be interesting to attempt to relax our restriction that the $p$-values must be independent. In the absence of some modification, our AdaPT procedure does not control FDR in finite samples for dependent $p$-values. In particular, there is a danger of ``overfitting'' to local random effects shared by nearby hypotheses: to the AdaPT procedure, such random effects are treated as signal to discover.

It could be interesting to pursue a hybrid method using ideas from AdaPT and Knockoff+ procedures in the case where the $p$-values arise from regression coefficients or other multivariate Gaussian test statistics. Suppose that we observe feature matrix $X\in\R^{n\times d}$ and response vector $y \sim \cN_n(X\beta, \sigma^2 I_n)$, and we wish to test hypotheses $H_j:\; \beta_j=0$ for $j=1,\ldots,d$. The key step in \citet{barber2015controlling} is to compute another matrix $\tX\in \R^{n\times d}$ with $\tX'\tX = X'X$ and $\tX'X = X'X - D$, for some diagonal $D\in \R^{d\times d}$ with positive entries; this can be done provided that $n \geq 2d$ and $X$ has full column rank.

If we define $v = X'y$ and $\tv = \tX'y$, then we have
\[
\begin{bmatrix} v+\tv \\ v-\tv \end{bmatrix} \sim \cN_d\left(\begin{bmatrix} (2X'X - D)\beta \\ D\beta \end{bmatrix}, \sigma^2\begin{bmatrix}4X'X - 2D & 0\\ 0 & 2D \end{bmatrix} \right).
\]
As a result $((v_j,\tv_j))_{j\in \cH_0}$ are independent exchangeable pairs, conditional on $(v_j)_{j\notin\cH_0}$. Let $\cF_{-1}= \sigma((\{v_j,\tv_j\})_{j=1}^d)$. The knockoff filter directly uses these exchangeable pairs by constructing knockoff statistics $w(X,y)\in \R^d$. The sufficiency and antisymmetry conditions together imply that each $|w_j|$ is $\cF_{-1}$-measurable and that, conditional on $\cF_{-1}$, $b_j = 1-\text{sgn}(w_j)$ is a mirror-conservative ``binary $p$-value:'' that is $(b_j)_{j\in\cH_0}$ are i.i.d. $\text{Bern}(1/2)$ independently of $\cF_{-1}$ and $(b_j)_{j\notin\cH_0}$. Using $|w_j|$ as a ``predictor'' (along with any other predictors for feature $j$ that we might have at hand) and $b_j$ as the $p$-value, the AdaPT procedure is immediately applicable. 

Note that $\min\{b_j,1-b_j\} = 0$ for every $j$; hence, at each step it matters only where the rejection threshold surface is above zero or not. If $q_t$ is the $t$th smallest value of $(|w_j|)_{j=1}^d$, the Knockoff+ filter corresponds to using the thresholds $s_t(|w_j|) = 0.5 \cdot 1\{|w_j| \geq q_t\}$. More generally, we can use AdaPT and interactively change the threshold we use.
 
If $\sigma^2$ is known, we can proceed more directly by constructing $z$-statistics and two-tailed $p$-values:
\[
z_j = \frac{v_j-\tv_j}{\sqrt{2d_j\sigma^2}} \sim \cN\left(\frac{2\beta_j}{\sqrt{2d_j\sigma^2}}, 1\right); \quad p_j = 2\min\{\Phi(z_j), 1-\Phi(z_j)\}.
\]
In that case $(p_j)_{j\in\cH_0}$ are i.i.d. uniform $p$-values conditional on $(p_j)_{j\notin\cH_0}$ and $v+\tv$ (not on $\cF_{-1}$ above). Once again, we can immediately apply AdaPT using  $v+\tv$ as a ``predictor.'' While it is not fully clear {\em a priori} just how we should use $v+\tv$ as a predictor, this represents an interesting avenue for future work.

\subsection{Connection to knockoffs in the orthogonal design case}\label{sec:knockoffs}

Focusing on the case of orthogonal design further illuminates the relationship between AdaPT and the Knockoff+ procedure. Suppose that $X \in \R^{n\times d}$ has orthonormal columns, and that $d \geq 2n$. In that case \citet{barber2015controlling} suggest using the knockoff matrix $\tX$ of $d$ more orthonormal columns which are also orthogonal to the columns of $X$. Then $X'y \sim \cN_d(\beta, \sigma^2 I_d)$ while $\tX'y \sim \cN_d(0,\sigma^2 I_d)$, independently. 

In this case, using the LASSO, forward stepwise regression, or virtually any other model selection path procedure on the design matrix $[X \tX]$ is identical to selecting variables in decreasing order of absolute value of $|X_j'y|$ and $|\tX_j'y|$; or equivalently, in increasing order of the two-tailed $p$-values $p_j = 2 - 2\Phi(|X_j'y|/\sigma)$ and $p_j^* = 2 - 2\Phi(|\tX_j'y|/\sigma)$ (this is true whether or not $\sigma^2$ is known). As a result, if we operationalize the Knockoff+ procedure using e.g. LASSO, we would reject hypotheses $H_j$ for which $\min\{p_j,p_j^*\}$ is small {\em and} $p_j < p_j^*$. By contrast, if we were to implement AdaPT with a constant threshold in each step, we would reject hypotheses $H_j$ for which $\min\{p_j,1-p_j\}$ is small {\em and} $p_j < 1-p_j$. Hence, the pairwise exchangeability of $(p_j,1-p_j)$ is playing the same role as the i.i.d. pair $(p_j,p_j^*)$ in knockoffs.

The two most salient differences between AdaPT and Knockoff+ in this case are that:
\begin{enumerate}
\item AdaPT allows for iterative interaction between the analyst and data, allowing the analyst to update her local FDR estimates as information accrues. By contrast, the knockoff filter as described in \citet{barber2015controlling} does not allow for such interaction (though it could, and this is a potentially interesting avenue for extending knockoffs).
\item Unlike Knockoff+, AdaPT introduces no extra randomness into the problem. This is because AdaPT uses pairwise exchangeability of $p_i$ with the ``mirror image'' $p$-value $1-p_i$ instead of the independent ``knockoff'' $p$-value $p_i^*\sim U[0,1]$. Thus, as a statistical procedure AdaPT respects the sufficiency principle: for any (non-randomized) choice of \update subroutine, the AdaPT result is a deterministic function of the original data.
\end{enumerate}

\subsection{Extension: estimating local FDR}

In addition to returning a list of rejections that is guaranteed to control the global FDR, most implementations of AdaPT will also return estimates, for each rejected hypothesis, of the local FDR,
\[
\hfdr(p_i \mid x_i) = \widehat{\P}(H_i \text{ is null} \mid x_i, p_i).
\]

If we have reasonably high confidence in the model we have used to produce these estimates, they may provide the best summary of evidence against the individual hypothesis $H_i$. By contrast, the significance level for global FDR only summarizes the strength of evidence against the entire list of rejections, taken as a whole. Indeed, it is possible to construct pathological examples where $\fdr(p_i \mid x_i) = 1$ for some of the rejected $H_i$, despite controlling FDR at some level $\alpha \ll 1$. Even apart from such perversities, it will typically be the case that $\hfdr(p_i \mid x_i) > \alpha$ for many of the rejected hypotheses. 

Despite their more favorable interpretation, however, the local FDR estimates produced by AdaPT rely on much stronger assumptions than the global FDR control guarantee --- namely, that the two-groups model, as well as our specifications for $\pi_1(x)$ and $f_1(p \mid x)$, must be correct. Instead of using the parametric estimates $\hfdr(p_i \mid x_i)$, we could estimate the local FDR in a moving window of $w$ steps of the AdaPT algorithm:
\[
\hfdp_{t,w} = \frac{A_t - A_{t+w}}{1 \vee (R_t - R_{t+w})}, \quad \text{ or } \hfdp_{t,w}^+ = \frac{1 + A_t - A_{t+w}}{1 \vee (R_t - R_{t+w})}.
\]
Note that if we take an infinitely large window, we obtain $\hfdp_{t, \infty}^+ = \hFDP_t$; thus, these estimators adaptively estimate the false discovery proportion for $p$-values revealed in the next $w$ steps of the algorithm, in much the same way that $\hFDP_t$ estimates the false discovery proportion for {\em all} remaining $p$-values. It would be interesting to investigate, in future work, what error-control guarantees we might be able to derive by using these estimators.

\section*{Acknowledgments}

The authors thank Jim Pitman, Ruth Heller, Aaditya Ramdas, and Stefan Wager for helpful discussions.

\bibliography{biblio}
\bibliographystyle{chicago}

\appendix

\section{EM algorithm details}\label{app:EM}

\subsection{Derivation of E-step}\label{subapp:Estep}
To fill in the details of the algorithm, we are left to calculate the imputed values $\hat{H}_i^{(r)}$ and $\hat{y}_{i}^{(r, 1)}$ given the parameters $\theta$ and $\beta$. Denote $\pi_{i}$ and $\mu_{i}$ by
\[\pi_{i} = \lb 1 + e^{-\phi_{\pi}(x_i)'\theta}\rb^{-1}, \quad \mu_{i} = \zeta^{-1}\lb \beta'\phi_{\mu}(x_i)\rb\]
We distinguish two cases: for revealed p-values, $\td{p}_{t, i}$ is a singleton; for masked p-values, $\td{p}_{t, i}$ is a two-elements set. For clarity, we let $p_{t, i}'$ denote the minimum element of $\td{p}_{t, i}$, i.e. $\td{p}_{t, i} = p_{t,i}'$ in the former case and $\td{p}_{t, i} = \{p_{t,i}', 1 - p_{t,i}'\}$ in the latter case.

\begin{itemize}
\item For revealed p-values,
\begin{align}
\hat{H}_{i}^{(r)} &= \P(H_{i} = 1 | p_{i} = \td{p}_{t, i}) = \frac{\pi_{i} \cdot h(\td{p}_{t,i}; \mu_{i})}{\pi_{i}\cdot h(\td{p}_{t,i}; \mu_{i}) + 1 - \pi_{i}},\label{eq:Ezi_case1}
\end{align}
and 
\begin{align}
\hat{y}_{i}^{(r, 1)} = \E (y_{i}H_{i} | \td{p}_{t, i}) / \hat{H}_{i}^{(r)} = y_{i}.
\end{align}
\item For masked p-values,
\begin{equation}\label{eq:pi1}
\P(p_{i} = p_{t, i}' | \td{p}_{t, i}, H_{i} = 1) = \frac{h(p_{t, i}'; \mu_{i})}{h(p_{t, i}'; \mu_{i}) + h(1 - p_{t, i}'; \mu_{i})}
\end{equation}

and 
\begin{equation}\label{eq:pi2}
\P(p_{i} = 1 - p_{t, i}' | \td{p}_{t, i}, H_{i} = 1) = \frac{h(1 - p_{t, i}'; \mu_{i})}{h(p_{t, i}'; \mu_{i}) + h(1 - p_{t, i}'; \mu_{i})}
\end{equation}

By Bayes' formula, we can also derive the conditional distribution of $H_{i}$ given $\td{p}_{t, i}$:
\begin{equation}\label{eq:Ezi_case2}
   \hat{H}_{i}^{(r)} = \P(H_{i} = 1| \td{p}_{t, i}) = \frac{\pi_{i} \cdot \lb h(p_{t, i}'; \mu_{i}) + h(1 - p_{t, i}'; \mu_{i})\rb}{\pi_{i} \cdot \lb h(p_{t, i}'; \mu_{i}) + h(1 - p_{t, i}'; \mu_{i})\rb + 2(1 - \pi_{i})}.
\end{equation}

As a consequence of \eqref{eq:pi1} - \eqref{eq:Ezi_case2}, 
\begin{align}\label{eq:Ezipi_case2}
  \hat{y}_{i}^{(r, 1)} = &\frac{1}{\hat{H}_{i}^{(r)}}\cdot \{ g(p_{t, i}') \cdot \P(H_{i} = 1, p_{i} = p_{t, i}' | \td{p}_{t, i}) + g(1 - p_{t, i}')\cdot \P(H_{i} = 1, p_{i} = 1 - p_{t, i}' | \td{p}_{t, i})\}\nonumber\\  
= & \frac{1}{\hat{H}_{i}^{(r)}}\cdot \frac{\pi_{i} \cdot \lb h(p_{t,i}';\mu_{i}) \cdot g(p_{t,i}') + h(1 - p_{t, i}'; \mu_{i})\cdot g(1 - p_{t,i}')\rb}{\pi_{i} \cdot \lb h(p_{t, i}'; \mu_{i}) + h(1 - p_{t, i}'; \mu_{i})\rb + 2(1 - \pi_{i})}\nonumber\\
= & \frac{h(p_{t,i}';\mu_{i}) \cdot g(p_{t, i}') + h(1 - p_{t, i}'; \mu_{i})\cdot g(1 - p_{t, i}')}{h(p_{t, i}'; \mu_{i}) + h(1 - p_{t, i}'; \mu_{i})}.
\end{align}
\end{itemize}

\subsubsection{Beta-mixture model}\label{subsubapp:GammaGLM}
Consider the beta-mixture model \eqref{eq:beta_mixture}, where
\[h(p; \mu) = \frac{1}{\mu}\cdot p^{\frac{1}{\mu} - 1}.\]
Plug it into our general results, we obtain that
\begin{itemize}
\item for revealed p-values,
  \begin{equation}\label{eq:Gamma_GLM_case1}
\hat{H}_{i}^{(r)} = \frac{\pi_{i} \cdot \frac{1}{\mu_{i}}\td{p}_{t, i}^{\frac{1}{\mu_{i}} - 1}}{\pi_{i}\cdot \frac{1}{\mu_{i}}\td{p}_{t, i}^{\frac{1}{\mu_{i}} - 1} + 1 - \pi_{i}}, \quad \hat{y}_{i}^{(r)} = y_i;
\end{equation}
\item for masked p-values,
\begin{align}
& \hat{H}_{i}^{(r)} = \frac{\pi_{i} \cdot \frac{1}{\mu_{i}}\lb p_{t, i}'^{\frac{1}{\mu_{i}} - 1} + (1 - p_{t, i}')^{\frac{1}{\mu_{i}} - 1}\rb}{\pi_{i} \cdot \frac{1}{\mu_{i}}\lb p_{t, i}'^{\frac{1}{\mu_{i}} - 1} + (1 - p_{t, i}')^{\frac{1}{\mu_{i}} - 1}\rb + 2(1 - \pi_{i})}, \nonumber\\
& \hat{y}_{i}^{(r)} = \frac{p_{t, i}'^{\frac{1}{\mu_{i}} - 1}(-\log p_{t, i}') + (1 - p_{t, i}')^{\frac{1}{\mu_{i}} - 1}(-\log (1 - p_{t, i}'))}{p_{t, i}'^{\frac{1}{\mu_{i}} - 1} + (1 - p_{t, i}')^{\frac{1}{\mu_{i}} - 1}}.\label{eq:Gamma_GLM_case2}
\end{align}
\end{itemize}

\subsubsection{Gaussian-mixture model}
Suppose the p-values are derived from a one-sided z-test by transforming a normal random variable, the most natural transformation is $g(p) = \Phi^{-1}(1 - p)$. By \eqref{eq:exp_family2},
\[h(p; \mu) = \exps{\mu \cdot g(p) - \frac{1}{2}\mu^{2}}.\]
Plug it into our general results, we obtain that
\begin{itemize}
\item for revealed p-values,
\begin{align}
\hat{H}_{i}^{(r)} &= \P(H_{i} = 1 | p_{i} = \td{p}_{t, i})  = \frac{\pi_{i} \cdot e^{\mu_i y_i}}{\pi_{i}\cdot e^{\mu_i y_i} + (1 - \pi_{i}) e^{\mu_i^2 / 2}},\quad \hat{y}_{i}^{(r)} = y_i; \label{eq:Gaussian_GLM_case1}
\end{align}
\item for masked p-values,
\begin{align}
& \hat{H}_{i}^{(r)} = \frac{\pi_{i} \cdot  \cosh(\mu_i y_i)}{\pi_{i} \cdot \cosh(\mu_i y_i) + (1 - \pi_{i})e^{\mu_i^2 / 2}}, \quad \hat{y}_{i}^{(r)} = |y_i|\cdot \tanh(\mu_i|y_i|).\label{eq:Gaussian_GLM_case2}
\end{align}
\end{itemize}

\subsection{Initialization}\label{subapp:EM_init}

Another important issue is the initialization. The formulae of $\hat{H}_{i}^{(r)}$ and $\hat{y}_{i}^{(r, 1)}$ requires estimates of $\pi_{1i}$ and $\mu_{i}$. In step 0 when no information can be obtained, we propose a simple method by imputing random guess as follows. First we obtain an initial guess of $\pi_{1i}$. Let $J_{i} = I(\td{p}_{t, i}\mbox{ contains two elements})$, then we observe that
\[\E J_{i} = \P(p_{i}\not\in [s_{0}(x_{i}), 1 - s_{0}(x_{i})]) \ge (1 - \pi_{1i})(1 - 2s_{0}(x_{i}))\]
\begin{equation}\label{eq:tdJi1}
\Longrightarrow \pi_{1i}\ge \E \lb 1 - \frac{J_{i}}{1 - 2s_{0}(x_{i})}\rb.
\end{equation}
Let 
\begin{equation}\label{eq:tdJi2}
\td{J}_{i} = 1 - \frac{J_{i}}{1 - 2s_{0}(x_{i})}
\end{equation}
and then we fit a logistic regression on $\td{J}_{i}$ with covariates $\phi_{\pi}(x_{i})$, denoted by $\hat{\pi}_{1i}$ and then truncate $\hat{\pi}_{1i}$ at $0$ and $1$ to obtain an initial guess of $\pi_{1i}$, i.e.
\begin{equation}
  \label{eq:pi0}
  \pi_{1i}^{(0)} = (\hat{\pi}_{1i}\vee 0)\wedge 1.
\end{equation}
\eqref{eq:tdJi1} implies that $\pi_{1i}^{(0)}$ is a conservative estimate of $\pi_{1i}$. This is preferred to an anti-conservative estimate since the latter might cause over-fitting. 

Then we obtain an initial guess of $\mu(x_{i})$ by imputing $p_{i}$'s. If $p_{i}\in [s_{0}(x_{i}), 1 - s_{0}(x_{i})]$, then $\td{p}_{t, i} = p_{i}$ and hence we can use it directly. Otherwise, we only know that $p_{i}\in \td{p}_{t, i} = \{p_{t, i}', 1 - p_{t,i}'\}$. If $p_{i}$ is null, then it should be uniform on $\{p_{t, i}', 1 - p_{t, i}'\}$; if $p_{i}$ is non-null, then it should more likely to be $p_{t, i}'$ since $p_{t, i}' < 1 - p_{t, i}'$. Thus, we impute $p_{i}$ by $p_{t, i}'$, and fit an unweighted GLM on $g(p'_{t, i})$ with covariates $\phi(x_{i})$ and inverse link to obtain an initial guess of $\mu_{i}$.

\subsection{Other issues}
In Algorithm \ref{algo:em}, we fit $\mu(x)$ using a weighted GLM, corresponding to the the M-step. However, the weighted step is sensitive to the weights $\hat{H}_{r}$ derived from the E-step, which relies on the assumption that null p-values are uniformly distributed on $[0, 1]$. In practice, null p-values might be super-uniform or only asymptotically uniform. In this case, the weights generated by the E-step might lead to abnormal estimates for $\mu(x)$. For this reason, we modify the weighted step in M-steps for fitting $\mu(x)$ into an unweighted step and find that this choice leads to consistently good performance in all experiments shown in Section \ref{sec:experiments}.

\section{Technical Proofs}
\begin{proof}[\textbf{Proof of Theorem \ref{thm:level_surface}}]
Assume $f_{0}(p \mid x_{i})\le M$. Let $\eta_{i} = \nu(\{x_{i}\})$ and $s = (s_{1}, \ldots, s_{n}) \triangleq (s(x_{1}), \ldots, s(x_{n}))$, then the objective function of \eqref{eq:np_simplify}
\[\int_{\mathcal{X}}-F_{1}(s(x) | x)\pi_1(x)\nu(dx) = -\sum_{i=1}^{n}\eta_{i}F_{1}(s_{i} | x_{i})\pi_1(x_{i})\]
is a convex function of $s$ by condition (i) and the constraint function 
\begin{align*}
&g(s)\triangleq \int_{\mathcal{X}}\bigg\{-\alpha F_{1}(s(x) | x)\pi_1(x) + (1 - \alpha)F_{0}(s(x) | x)(1 - \pi_1(x))\bigg\}\nu(dx)\\
=& \sum_{i=1}^{n}\eta_{i}\lb -\alpha F_{1}(s_{i} | x_{i})\pi_1(x_{i}) + (1 - \alpha)F_{0}(s_{i} | x_{i})(1 - \pi_1(x_{i}))\rb
\end{align*}
is also a convex function of $s$ by condition (i). To establish the necessity of the KKT condition, it is left to prove the Slater's condition (\citealt{slater50}; \citealt[Chap. 5]{boyd04}), i.e. there exists a $\bar{s}$, such that for any $s\in B(\bar{s}, \delta)$ for some $\delta > 0$, the constraint inequality holds, i.e. $g(s)\le 0$, and $g(\bar{s}) < 0$. By condition (ii), WLOG we assume $\fdr(0 | x_{1}) < \alpha$ with $\nu(\{x_{1}\}) = \eta_{1} > 0$. Fix any $\eps > 0$ and denote by $\omega(\eps)$ by the maximum modulus of continuity of $f_1(p \mid x_1)$ and $f_0(p \mid x_0)$ at point $0$, i.e.
\[\omega(\eps) = \max\left\{\sup_{p\le\eps} f_1(p \mid x_1), \sup_{p\le\eps} f_0(p \mid x_1) \right\}.\]
By assumption (i), we know that
\[\lim_{\eps\rightarrow 0}\omega(\eps) = 0.\]
Let $\bar{s}\in\R^{n}$ with 
\[\bar{s}_{1} = 2\eps, \bar{s}_{2} = \ldots = \bar{s}_{n} = \eps \cdot (2M)^{-1}\eta_{1}\Delta\]
where 
\[\Delta = f(0 | x_{1})\lb \alpha (1 - \fdr(0 | x_{1})) - (1 - \alpha)\fdr(0 | x_{1})\rb > 0.\]
Let $\delta = \min\{\bar{s}_{2}, \eps\}$. We will show that for any $s\in B(\bar{s}, \delta)$, $g(s) < 0$. In fact, for any $s \in B(\bar{s}, \delta)$, we have
\[s_{1}\in [\eps, 3\eps], \quad s_{i}\le \eps\cdot M^{-1}\eta_{1}\Delta.\]
Recalling that $M$ is an upper bound for the null density. WLOG, we assume that $M \ge \eta_1\Delta$ in which case $s_{i}\le \eps$ for all $i\ge 2$. Note that $g(0) = 0$. By mean-value theorem, there exists $\td{s}_{1} \in [0, 3\eps], \td{s}_{2}, \ldots, \td{s}_{n}\in [0, \eps]$, such that 
\begin{align*}
g(s) & = s_{1}\eta_{1}\lb -\alpha f_{1}(\td{s}_{1} | x_{1})\pi_1(x_{1}) + (1 - \alpha)f_{0}(\td{s}_{1} | x_{1})(1 - \pi_1(x_{1}))\rb\\
&\quad  + \sum_{i=2}^{n}s_{i}\eta_{i}\lb -\alpha f_{1}(\td{s}_{i} | x_{i})\pi_1(x_{i}) + (1 - \alpha)f_{0}(\td{s}_{i} | x_{i})(1 - \pi_1(x_{i}))\rb\\
& \le s_{1}\eta_{1}\lb -\alpha f_{1}(0 | x_{1})\pi_1(x_{1}) + (1 - \alpha)f_{0}(0 | x_{1})(1 - \pi_1(x_{1}))\rb + s_1 \eta_1 \omega(\eps)\\
&\quad  + \sum_{i=2}^{n}s_{i}\eta_{i}\lb -\alpha f_{1}(\td{s}_{i} | x_{i})\pi_1(x_{i}) + (1 - \alpha)f_{0}(\td{s}_{i} | x_{i})(1 - \pi_1(x_{i}))\rb\\
& = -s_{1}\eta_{1}\Delta + s_1 \eta_1 \omega(\eps) + \sum_{i=2}^{n}s_{i}\eta_{i}\lb -\alpha f_{1}(\td{s}_{i} | x_{i})\pi_1(x_{i}) + (1 - \alpha)f_{0}(\td{s}_{i} | x_{i})(1 - \pi_1(x_{i}))\rb\\
& \le -s_{1}\eta_{1}\Delta + s_1 \eta_1 \omega(\eps) + \sum_{i=2}^{n}s_{i}\eta_{i}\cdot M\\
& \le -\eps\eta_{1}\Delta + s_1 \eta_1 \omega(\eps) + \eps \eta_{1}\Delta\sum_{i=2}^{n}\eta_{i}\\
& \le -\eps\eta_{1}\Delta + O(\eps\omega(\eps)) + \eps \eta_{1}\Delta(1 - \eta_1)\\
& \le -\eps\eta_{1}^{2}\Delta + o(\eps).
\end{align*}
Thus, for sufficiently small $\eps$, $g(s) < 0$ for all $s\in B(\bar{s}, \delta)$ and hence the Slater's condition is satisfied. 
\end{proof}

\begin{proof}[\textbf{Proof of Lemma \ref{lem:bernoulli}}]
  We assume $\rho<1$ (otherwise the result is trivial). Following \citet{barber2016knockoff}, we introduce the random set $\cA \subseteq [n]$ with
  \[
  \P(i \in \cA \mid \cG_{-1}) = \frac{1-\rho_i}{1-\rho},
  \]
  conditionally independent for $i\in [n]$, and construct conditionally i.i.d. Bernoulli variables $q_1,\ldots,q_n$, independent of $\cA$, with ${\P(q_i = 1 \mid \cG_{-1}) = \rho}$. Then we can define
\begin{equation}\label{eq:tdbi}
  \td{b}_i = q_i \1\{i \in \cA\} + \1\{i \notin \cA\},
\end{equation}
  which by construction gives $\P(\td{b}_i = 1 \mid \cG_{-1}) = \rho_i$ almost surely. Furthermore, noticing that 
\begin{align*}
&\P(\td{b}_{i} = 0, \td{b}_{j} = 0 | \cG_{-1}) = \P(i\in \cA, j\in \cA, q_{i}= 0, q_{j} = 0 | \cG_{-1}) \\
= &\P(i\in \cA, q_{i}= 0 | \cG_{-1})\P(j\in \cA, q_{j} = 0| \cG_{-1}) \\
= &\P(\td{b}_{i} = 0|\cG_{-1})\P(\td{b}_{j} = 0 | \cG_{-1}),
\end{align*}
we conclude that the $\td{b}_i$ are conditionally independent given $\cG_{-1}$. As a consequence, given $\cG_{-1}$,
\[
  (\td{b}_{1}, \ldots, \td{b}_{n})\stackrel{d}{=} (b_{1}, \ldots, b_{n}).
\]
In the following proof, we will use \eqref{eq:tdbi} to represent $b_{i}$'s.

  To ensure that $\cC_t$ decreases by at most a single element in each step, we introduce intermediate steps: for integers $t\geq 0$, $1\leq i \leq n$ define 
  \[
  \cC_{t + i/n} = \cC_{t + 1} \cup \{j \leq n-i:\; j \in \cC_{t}\}.
  \]

  Next, define the augmented filtration
  \[
  \cG_t^{\cA} = \sigma\left(\cG_{-1}, \cA, \cC_t, (b_i)_{i \notin \cC_t \cap \cA}, \sum_{i \in \cC_t \cap \cA} b_i\right) \supseteq \cG_t,
  \]
  for both integer and fractional values of $t$. Note $\cC_{t+1/n}$ is measurable with respect to $\cC_t$. In addition we define

  \[
  U_t^{\cA} = \sum_{i \in \cC_t \cap \cA} b_i, \qquad
  V_t^{\cA} = \sum_{i \in \cC_t \cap \cA} 1-b_i, \quad \text{ and } \quad
  Z^{\cA}_t = \frac{1 + |\cC_t\cap \cA|}{1+U_t^{\cA}}\;.
  \]
  Recall the definition of $U_{t}$ (Section \ref{subsec:notation}) and $b_{i}$ (defined above in \eqref{eq:tdbi}), for any $t$,
\[\frac{1 + |\cC_t|}{1 + U_{t}} = \frac{1 + |\cC_t\cap \cA| + |\cC_t\cap \cA^{c}|}{1 + U_{t}^{\cA} + |\cC_t\cap \cA^{c}|}\le \frac{1 + |\cC_t\cap \cA|}{1 + U_{t}^{\cA}} = Z^{\cA}_{t}.\]
Finally, we observe that $(b_i)_{i \in \cC_t \cap  \cA} = (q_i)_{i \in \cC_t \cap \cA}$ are exchangeable with respect to $\cG_t^{\cA}$, with the random vector distributed uniformly over configurations summing to $U_t^{\cA}$. 

  There are three cases: 
  \begin{enumerate}[(i)]
  \item if $\cC_{t+1/n}\cap \cA = \cC_t\cap \cA$ then 
\[\E[Z^{\cA}_{t+1/n} \mid \cG_t^{\cA}] = Z^{\cA}_t;\]
  \item if $U_t^{\cA} = 0$ but $\cC_{t+1/n}\cap \cA \subsetneqq \cC_t\cap \cA$ then 
\[Z^{\cA}_{t+1/n} = 1 + |\cC_{t + 1/n}\cap \cA| \leq |\cC_{t}\cap \cA| =Z^{\cA}_{t} - 1\le Z^{\cA}_{t};\]
  \item otherwise, $U_{t}^{\cA} > 0$ and $\cC_{t} \setminus \cC_{t + 1/n} = \{j\}$ for some $j\in \cA$. The exchangeability of $b_{i} = q_{i}$ implies that 
\[P\lb b_{j} = 1 \mid \cG_{t}^{\cA}\rb = \frac{U_{t}^{\cA}}{U_{t}^{\cA} + V_{t}^{\cA}}.\] 
Then
  \begin{align*}
    \E\left[Z^{\cA}_{t+1/n} \mid \cG_t^{\cA} \right]
    &= \frac{U_{t}^{\cA} + V_{t}^{\cA}}{1 + U_{t}^{\cA}} \,\cdot\, \frac{V_{t}^{\cA}}{U_{t}^{\cA}+V_{t}^{\cA}} \;+\; 
    \frac{U_{t}^{\cA} + V_{t}^{\cA}}{U_{t}^{\cA}} \,\cdot\, \frac{U_{t}^{\cA}}{U_{t}^{\cA}+V_{t}^{\cA}} \\
    &= \frac{V_{t}^{\cA}}{1 + U_{t}^{\cA}} + 1  \; = Z^{\cA}_t \;.
  \end{align*}
  \end{enumerate}
  In all three cases, the conditional expectation of $Z^{\cA}_{t+1/n}$ is smaller than $Z^{\cA}_t$; thus, $Z^{\cA}_t$ is a super-martingale with respect to the filtration $\cG^{\cA}_t$. Because $\htt$ is also a stopping time with respect to the filtration $(\cG_t^{\cA})_{t=0,1/n,2/n,\ldots}$ (but one which can only take integer values), for any $\cA\in [n]$, we have
  \begin{align}
    \E\left[\frac{ 1 + |\cC_{\htt}| }{1 + \sum_{i\in \cC_{\htt}} b_i} \middle| \cG_{-1}, \cA\right] \leq  \E\left[Z^{\cA}_{\htt} \mid \cG_{-1}, \cA\right] \leq  \E\left[Z^{\cA}_0 \mid \cG_{-1}, \cA\right] = \E\left[Z^{\cA}_0 \mid \cG_{-1}\right].\label{eq:Z0A}
  \end{align}
Let $m = |\cC_{0}|$ and assume $\cC_{0} = \{1,\ldots, m\}$ WLOG. Using the representation \eqref{eq:tdbi}, 
\begin{align*}
\E\left[Z^{\cA}_0 \mid \cG_{-1}\right] &= \E \left[\frac{1 + m}{1 + \sum_{i=1}^{m}(q_i I(i\in \cA) + I(i\not\in \cA))} \mid \cG_{-1}\right]\\
& \le \E \left[\frac{1 + m}{1 + \sum_{i=1}^{m}q_i } \mid \cG_{-1}\right] = \E \left[\frac{1 + m}{1 + \sum_{i=1}^{m}q_i }\right]\\
& = \sum_{k=0}^{m}\frac{1 + m}{1 + k}\cdot \lb\begin{array}{cc}m\\k\end{array}\rb \rho^{k}(1 - \rho)^{m-k}\\
& = \sum_{k=0}^{m}\lb\begin{array}{cc}m + 1\\k + 1\end{array}\rb \rho^{k}(1 - \rho)^{m-k}\\
& = \rho^{-1}\cdot \sum_{k=0}^{m}\lb\begin{array}{cc}m + 1\\k + 1\end{array}\rb \rho^{k + 1}(1 - \rho)^{m + 1 - (k + 1)}\\
& = \rho^{-1}\lb 1 - (1 - \rho)^{m + 1}\rb \le \rho^{-1}.
\end{align*}
Marginalizing over $\cA$ in \eqref{eq:Z0A}, we obtain the result.
\end{proof}

\subsection{Mirror-conservatism}\label{subapp:mirror}
In this subsection we provide two important examples that produce mirror-conservative p-values. The first example is the permutation test (e.g. ~\cite{hoeffding1952large}). Typically we assume that under the null hypothesis, the test statistic $T(X)$, where $X$ is a short-handed notation for observed data, is invariant in distribution under a finite group of transformations $\mathcal{G}$, i.e.
\[T(X)\stackrel{d}{=}T(gX), \forall g\in \mathcal{G}.\]
When $|\mathcal{G}|$ is small, one can compute a discrete p-value by $R / |\mathcal{G}|$, where $R$ is the rank of $T(X)$ in the set $\{T(gX): g\in \mathcal{G}\}$. When $|\mathcal{G}|$ is large, \cite{hemerik2014exact} proposes sampling a subset $\mathcal{G}' = \{g_{1}, g_{2}, \ldots, g_{m}\}$ with $g_{1} = \mathrm{Id}$ and $g_{2}, \ldots, g_{m}$ being a simple random sample (without replacement) from $\mathcal{G}$ and calculate the p-value based on $\mathcal{G}'$. In both cases, it can be proved that the p-value is uniformly distributed on an equi-spaced grid $\{\frac{1}{m}, \frac{2}{m},\ldots, \frac{m}{m}\}$ under the null hypothesis, where $m$ is the number of replicates. Then for any $0 < a_{1}\le a_{2} \le \frac{1}{2}$,
\[P(p\in [a_{1}, a_{2}]) = \lceil m a_{2}\rceil - (\lfloor m a_{1}\rfloor - 1)_{+}\]
where $(u)_{+}$ denotes $\max\{u, 0\}$, and 
\[P(p\in [1 - a_{2}, 1 - a_{1}]) = \lceil m (1 - a_1)\rceil - \lfloor m (1 - a_2)\rfloor + 1 = \lceil m a_{2}\rceil - (\lfloor m a_{1}\rfloor - 1).\]
As a result we conclude that 
\[P(p\in [a_{1}, a_{2}])\le P(p\in [1 - a_{2}, 1 - a_{1}])\]
and hence $p$ is mirror-conservative.

The second example is the one-sided test for distributions with monotone likelihood ratio, which is ubiquitous in practice. Specifically, let $\theta$ be the univariate parameter of interest and $p_{\theta}(x)$ be a family of densities with respect to some carrier measure $\mu$. $p_{\theta}(x)$ is said to have monotone likelihood ratio with respect to some real-value function $T(x)$ if for any $\theta < \theta'$, $p_{\theta}\not \equiv p_{\theta'}$ and the ratio $p_{\theta'}(x) / p_{\theta}(x)$ is a nondecreasing function of $T(x)$. For testing  $H_{0}: \theta\le \theta_{0}$ against $H_{1}: \theta > \theta_{0}$, it is well-known that there exists a Uniformly Most Powerful (UMP) test \citep{lehmann2005testing}, with the following decision function:
\[\phi(x) = \left\{
    \begin{array}{cc}
      1 & (T(x) > C)\\
      \gamma & (T(x) = C)\\
      0 & (T(x) < C)
    \end{array}
\right.\]
where $(\gamma, C)$ is the solution of 
\[P_{\theta_{0}}(T(X) > C) + \gamma P_{\theta_{0}}(T(X) = C) = \alpha.\]
Write $T(X)$ as $T$ and $P_{\theta_{0}}(T(X) \ge t)$ as $G_{0}(t)$ for short. Then the induced p-value can be written as
\begin{equation}\label{eq:monotone_likelihood}
p = G_{0}(T^{+}) + U(G_{0}(T) - G_{0}(T^{+})), \quad U\sim U([0, 1]),
\end{equation}
where $G_{0}(t^{+}) = \lim_{t\downarrow t^{+}} G_{0}(t)$. \eqref{eq:monotone_likelihood} is termed as fuzzy p-values by \cite{geyer2005fuzzy}. 
\begin{proposition}
Let $p_{\theta}(x)$ be a family of densities (w.r.t the carrier measure $\mu$) that has monotone likelihood ratio w.r.t. $T(x)$. 
Then the p-value defined in \eqref{eq:monotone_likelihood} is mirror-conservative.
\end{proposition}

\begin{proof}
Since $p_{\theta}(x)$ has monotone likelihood ratio, there exists  a non-decreasing function $g_{\theta}(t)$ for each $\theta\le \theta_{0}$, such that
\[\frac{p_{\theta}(x)}{p_{\theta_{0}}(x)} = g_{\theta}(T(x)).\]
Let $\nu$ be a measure such that for any event $A\subset \R$,
\[\nu(A) = \int I(T(x) \in A)\cdot p_{\theta_{0}}(x) \mu(dx).\]
Then for any event $A\subset \R$, 
\begin{equation}\label{eq:measure}
P_{\theta}(T(X) \in A) = \int g_{\theta}(T(x))\cdot I(T(x) \in A)\cdot p_{\theta_{0}}(x) d\mu = \int g_{\theta}(t) \nu (dt).
\end{equation}
Note that the above argument can be easily proved by standard approximation argument in measure theory that starts from indicator functions $g_{\theta}(t) = I(t\in A')$, extends the result to simple step functions and finally pushes it to the limit. Let $\omega$ be the product measure of $\nu$ and the Lebesgue measure on $[0, 1]$. Then for any event  $B\subset \R^{2}$,
\begin{equation}\label{eq:omega1}
P_{\theta}((T(X), U)\in B) = \int_{B}g_{\theta}(t) \omega(dt, du).
\end{equation}
Note that $g_{\theta_{0}}(t)\equiv 1$ by definition. This implies that 
\begin{equation}\label{eq:omega2}
\omega(B) = P_{\theta_{0}}((T(X), U)\in B).
\end{equation}

~\\
\noindent Let $H(\cdot, \cdot)$ be the transformation such that $p = H(T, U)$. Then for any $z$,
\[\{(t, u): G_{0}(t) < z\}\subset H^{-1}([0, z)) \subset H^{-1}([0, z]) \subset \{(t, u): G_{0}(t)\le z\}.\]
As a result, for any $0 < z_1 < z_2 < 1$,
\[H^{-1}([0, z_1]) \subset \{(t, u): G_{0}(t)\le z_1\}, \quad H^{-1}([z_2, 1]) \subset \{(t, u): G_{0}(t)\ge z_2\},\]
and hence there exists $t(z_1, z_2)$ such that
\begin{equation}\label{eq:separation}
t_1 \le t(z_1, z_2)\le t_2, \quad \forall (t_1, u_1)\in H^{-1}([0, z_1]), (t_2, u_2)\in H^{-1}([z_2, 1]).
\end{equation}
Given $0\le a_1 \le a_2 < 0.5$, let $A_1 = [a_1, a_2]$ and $A_2 = [1 - a_2, 1 - a_1]$. Then \eqref{eq:separation} and \eqref{eq:omega1}, together with the monotonicity of $g_{\theta}$, imply that\[P_{\theta}(p\in A_1) = \int_{H^{-1}(A_1)}g_{\theta}(t) \omega(dt, du)\le \omega(H^{-1}(A_1))\cdot g_{\theta}(t(a_2, 1-a_2)),\]
and 
\[P_{\theta}(p\in A_2) = \int_{H^{-1}(A_1)}g_{\theta}(t) \omega(dt, du)\ge \omega(H^{-1}(A_2))\cdot g_{\theta}(t(a_2, 1-a_2)).\]
Recalling \eqref{eq:omega2}, we obtain that 
\begin{equation}\label{eq:prob_ratio}
\frac{P_{\theta}(p\in A_1)}{P_{\theta}(p\in A_2)}\le \frac{P_{\theta_{0}}(p\in A_1)}{P_{\theta_{0}}(p\in A_2)}.
\end{equation}

~\\
\noindent It is left to prove that $\frac{P_{\theta_{0}}(p\in A_1)}{P_{\theta_{0}}(p\in A_2)} = 1$. In fact, we can prove that 
\begin{equation}\label{eq:uniform}
p\sim U([0, 1])\quad \mbox{when } \theta = \theta_{0}.
\end{equation}
Fix any $z\in (0, 1)$, let 
\[G_{0}^{-1}(z) = \sup\{t: G_{0}(t)\ge z\}.\]
For clarity we write $u$ for $G_{0}^{-1}(z)$. Now we prove \eqref{eq:uniform} in two cases:
\begin{itemize}
\item if $u$ is a continuity point of $G_{0}$, i.e.
\[G_{0}\lb u\rb = G_{0}\lb u^{+}\rb.\]
Since $G_{0}$ is left-continuous, we must have
\[G_{0}\lb u\rb = G_{0}\lb u^{+}\rb = z.\]
Then 
\[P_{\theta_{0}}\lb p\le z\rb = P_{\theta_{0}}\lb T(X)\ge u\rb = G_{0}\lb u\rb.\]
\item if $u$ is an atom of $G_{0}$, i.e.
\[G_{0}\lb u\rb > G_{0}\lb u^{+}\rb.\]
By definition,
\[G_{0}\lb u\rb\ge z\quad \mbox{and}\quad  z\ge G_{0}\lb u^{+}\rb.\]
Then 
\begin{align*}
  P_{\theta_{0}}\lb p\le z\rb & = P_{\theta_{0}}\lb T(X)> u^{+}\rb + P_{\theta_{0}}\lb T(X)= u\rb\cdot \frac{z - G_{0}(u^{+})}{G_{0}(u) - G_{0}(u^{+})}\\
& = G_{0}(u^{+}) + (G_{0}(u) - G_{0}(u^{+}))\cdot \frac{z - G_{0}(u^{+})}{G_{0}(u) - G_{0}(u^{+})}\\
& = z.
\end{align*}
\end{itemize}
Therefore we prove \eqref{eq:uniform}. By \eqref{eq:prob_ratio}, we conclude that for any $\theta\le \theta_{0}$,
\[\frac{P_{\theta}(p\in A_1)}{P_{\theta}(p\in A_2)}\le 1\]
which implies the mirror-conservativeness.

\end{proof}

\end{document}